\documentclass[conference]{IEEEtran}
\usepackage{amsmath}
\usepackage{amsthm}
\usepackage{amssymb}
\usepackage[caption=false,font=footnotesize]{subfig}
\usepackage{url}
\usepackage{mathtools}
\usepackage[utf8]{inputenc}
\usepackage[T1]{fontenc}
\usepackage{bussproofs}
\usepackage{tikz}
\usepackage{tikzit}
\usepackage{tikz-cd}
\usepackage{stmaryrd}
\usepackage{thmtools}
\usepackage{thm-restate}
\usepackage{hyperref}
\usepackage{cleveref}
\usepackage{cmll}

\declaretheorem[name=Theorem,numberwithin=section]{theorem}
\declaretheorem[name=Proposition,numberwithin=section,sibling=theorem]{proposition}
\declaretheorem[name=Lemma,numberwithin=section,sibling=theorem]{lemma}
\declaretheorem[name=Corollary,numberwithin=section,sibling=theorem]{corollary}

\declaretheorem[name=Definition,numberwithin=section,sibling=theorem,style=definition]{definition}
\declaretheorem[name=Example,numberwithin=section,sibling=theorem,style=definition]{example}

\declaretheorem[name=Remark,numberwithin=section,sibling=theorem,style=remark]{remark}

\input{tikzit.tikzdefs}

\tikzstyle{empty}=[shape=circle, tikzit fill={rgb,255: red,191; green,191; blue,191}]
\tikzstyle{dot}=[fill=black, draw=black, shape=circle]
\tikzstyle{hook}=[right hook->, draw=black, tikzit draw=magenta]
\tikzstyle{mono2}=[draw=black, >->]
\tikzstyle{epi2}=[draw=black, ->>]

\tikzstyle{to}=[draw=black, ->]
\tikzstyle{equal-arrow}=[-, double equal sign distance]
\tikzstyle{dashed-arrow}=[dashed, ->]
\tikzstyle{epi}=[draw=black, -, epi2]
\tikzstyle{mono}=[draw=black, -, mono2]
\tikzstyle{hook}=[draw=black, -, right hook->]

\title{Operator Spaces, Linear Logic and the Heisenberg-Schrödinger Duality of Quantum Theory}

\author{
  \IEEEauthorblockN{Bert Lindenhovius}
  \IEEEauthorblockA{
    \textit{Institute for Mathematical Methods in} \\
    \textit{Medicine and Data Based Modeling} \\
    \textit{Johannes Kepler University} \\
    Linz, Austria
  }
  \and
  \IEEEauthorblockN{Vladimir Zamdzhiev}
  \IEEEauthorblockA{
    \textit{Université Paris-Saclay, CNRS, ENS Paris-Saclay, Inria,} \\
    \textit{Laboratoire Méthodes Formelles,} \\
    91190, Gif-sur-Yvette, France
  }
}

\begin{document}

\maketitle

\begin{abstract}
  We show that the category $\OS$ of operator spaces, with complete
  contractions as morphisms, is locally countably presentable and a model of
  Intuitionistic Linear Logic in the sense of Lafont. We then describe a model
  of Classical Linear Logic, based on $\OS$, whose duality is compatible with
  the Heisenberg-\Schrod{} duality of quantum theory. We also show that $\OS$
  provides a good setting for studying pure state and mixed state quantum
  information, the interaction between the two, and even higher-order quantum
  maps such as the quantum switch.
\end{abstract}

\section{Introduction}
\label{sec:intro}

The Heisenberg-\Schrod{} duality establishes a duality
between two pictures of quantum theory. In the \Schrod{} picture, it is often
said that one modifies the state of the quantum system while keeping the
observable invariant. In the Heisenberg picture, it is the other way round, i.e.
the observable is modified while keeping the state invariant. There is a rich
mathematical theory, based on functional analysis \cite{pedersen:analysisnow}, noncommutative geometry \cite{connes:ncg}, and
operator algebras (e.g. von Neumann algebras) \cite{Blackadar17,kadisonringrose:oa1,kadisonringrose:oa2,takesaki}, that can
be used to describe this duality and that can make precise what we mean by
``state'' and ``observable''. 

What is missing from the picture (figuratively) is whether this duality can be
understood using Linear Logic \cite{linear-logic}. Indeed, Linear Logic (LL)
has been influential in our understanding of dualities, because it provides us
with a rich logical setting, equipped with numerous logical connectives, that
allow us to study dualities. This is reflected quite nicely via Polarised
Linear Logic \cite{pll} which makes it clear how certain LL formulas may be
assigned positive or negative logical polarities (see also \cite{pll-fixpoints}
for more information). Furthermore, the semantics of linear logic has been
extensively studied and the relevant categorical models are well-known
\cite{mellies-linear-logic,seely-linear,barr-autonomous-book} which gives us
another way to reason about the relevant dualities.

Our paper is concerned with the following question:
\begin{center}
  Can we use linear logic semantics to study the Heisenberg-\Schrod{} duality?
\end{center}
We provide strong evidence in support of an affirmative answer. In particular,
we construct a model of LL whose duality is compatible with the
Heisenberg-\Schrod{} duality and where system descriptions in the \Schrod{}
picture correspond to formulas with a positive logical polarity, whereas those
in the Heisenberg picture correspond to formulas with a negative logical
polarity. In the process of constructing this model, we prove mathematical
results in the intersection of noncommutative geometry and category theory, via
the theory of operator spaces \cite{effros-ruan,blecher-merdy,pisier}.

In Section \ref{sec:categorical-background}, we recall some preliminaries on
\emph{locally presentable categories} \cite{lp-categories-book}. This is a very
strong categorical property that is essential for establishing our results. The
main category on which we base our development is $\OS$, the category of
\emph{operator spaces} with complete contractions as morphisms. Operator spaces
are mathematical structures from the field of noncommutative geometry that are
widely recognised as the noncommutative (or quantum) generalisation of Banach
spaces. In Section \ref{sec:os}, we recall some preliminaries on operator
spaces, we recall that this is indeed a nice setting for the mathematical
description of the Heisenberg-\Schrod{} duality, and then we prove that $\OS$
is locally countably presentable. This, together with its structure as a
symmetric monoidal closed category, implies that it is a model of
Intuitionistic Linear Logic (ILL). In Section \ref{sec:pure-mixed}, we show
that $\OS$ can be used to study both pure state and mixed state quantum
information and the interaction between the two via semantic methods based on
ILL. We show that it is expressive enough to model higher-order quantum maps,
such as the quantum switch \cite{quantum-switch}, and we also showcase the use
of the \emph{Haagerup tensor product} \cite{effros-ruan} for such maps and its
potential interesting connections to BV-logic \cite{bv-logic}. In Section
\ref{sec:ll-hs}, we construct a model of Classical Linear Logic (CLL), i.e.
full linear logic, based on $\OS$ through the use of the Chu construction
\cite{Chu}. We show that the duality in the resulting model is compatible with
the Heisenberg-\Schrod{} duality and that it behaves well in a polarised sense:
von Neumann algebras (viewed as dual operator spaces) correspond to formulas
with negative logical polarity, whereas their predual operator spaces
correspond to formulas with positive logical polarity. In this polarised sense,
the multiplicative disjunction corresponds to the spatial tensor product of von
Neumann algebras, which is how systems are composed in the Heisenberg picture,
whereas the multiplicative conjunction corresponds to the monoidal tensor of
$\OS$ applied to their preduals, which is how systems are composed in the
\Schrod{} picture. In Section \ref{sec:discuss}, we discuss future work and
other related works.

\section{Categorical Preliminaries}
\label{sec:categorical-background}

In this section we recall some background on locally presentable categories.
They were originally introduced by Gabriel and Ulmer
\cite{gabriel-ulmer}. The background we present here is mostly based on
the textbook accounts in \cite{lp-categories-book,borceux2}.

\begin{definition}[$\alpha$-directed Poset]
  A partially ordered set (poset) $\Lambda$ is \emph{directed} whenever every finite
  subset $X \subseteq \Lambda$ has an upper bound in $\Lambda.$ If $\alpha$ is a regular
  cardinal, we say that a poset $\Lambda$ is \emph{$\alpha$-directed} if every
  subset $X \subseteq \Lambda,$ with cardinality strictly smaller than $\alpha$, has
  an upper bound in $\Lambda.$
\end{definition}

In particular, the notion of a directed set coincides with that of an
$\aleph_0$-directed set. We also say that an $\aleph_1$-directed set is
\emph{countably directed}. Indeed, a poset $\Lambda$ is $\aleph_1$-directed iff
every countable subset $X \subseteq \Lambda$ has an upper bound in $\Lambda.$


\begin{definition}[$\alpha$-directed Diagram]
  \label{def:countably-directed-colimit}
  An \emph{$\alpha$-directed diagram} in a category $\CC$ is a functor $D
  \colon \Lambda \to \CC,$ where $\Lambda$ is an $\alpha$-directed poset viewed
  as a posetal category in the obvious way.
  A colimiting cocone $(C, \{c_\lambda \colon D_\lambda \to C\}_{\lambda \in \Lambda} )$
  of $D$ is called an $\alpha$-\emph{directed colimit}.
  When $\alpha = \aleph_0$, we say that $D$ is a \emph{directed diagram}, when
  $\alpha = \aleph_1$ we say that $D$ is a \emph{countably-directed} diagram,
  and likewise for the colimit.
\end{definition}

An important concept in the theory of locally presentable categories is the
idea of an $\alpha$-presentable object that we define next. These objects are
particularly well-behaved in the theory.

\begin{definition}
  \label{def:presentable-object}
  Let $\alpha$ be a regular cardinal. An object $A$ of a category $\CC$ is
  \emph{$\alpha$-presentable} whenever the hom-functor $\CC(A,-) \colon \CC \to
  \Set$ preserves $\alpha$-directed colimits. When $\alpha = \aleph_0$ we say
  that $A$ is \emph{finitely-presentable} and when $\alpha = \aleph_1$ we say
  that $A$ is \emph{countably-presentable}.
\end{definition}

\begin{example}
  \label{ex:presentable-objects}
  In the category $\Set$, the finitely-presentable objects are the
  finite sets and the countably-presentable objects are the countable
  sets. In the category $\Vect_{\mathbb K}$ of vector spaces over a field
  $\mathbb K$, the finitely-presentable objects are the finite-dimensional
  vector spaces, whereas the countably-presentable objects are the
  countably-dimensional vector spaces \cite{finitary-functors}.
  In $\Ban$, the only finitely-presentable object is the
  zero-dimensional Banach space and the countably-presentable objects
  coincide with the separable ones \cite[Remark 2.7]{ban-separable}.
\end{example}

The most notable difference between the locally presentable structures of
$\Ban$ and $\OS$ is given by the \emph{strong generators} of the categories, a
notion that we recall next.

\begin{definition}
  A \emph{generating set} for a category $\CC$ is a (small) set of objects
  $\mathcal S \subseteq \Ob(\CC)$, such that for any pair of \emph{distinct}
  parallel morphisms $f, g \colon A \to B$ of $\CC$, there exists an object $S
  \in \mathcal S$ and a morphism $s \colon S \to A$ with the property that $f
  \circ s \neq g \circ s.$ A generating set $\mathcal S$ is \emph{strong}
  whenever the following condition holds: for any proper monomorphism $m \colon
  A \to B$ (i.e. a monomorphism that is not an isomorphism), there exists an
  object $S \in \mathcal S$ and a morphism $f \colon S \to B$ which does not
  factorise through $m$. An object $S$ is called a \emph{(strong) generator} if
  the singleton $\{ S \}$ is a (strong) generating family.
\end{definition}

\begin{example}
  \label{ex:strong-generators}
  In the category $\Set$, the terminal object $1$ (any singleton set) is a
  strong generator. In the category $\Vect_{\mathbb K}$, the field $\mathbb K$
  is a strong generator. In the category $\mathbf{Top}$, the terminal object 1
  (a singleton set with the unique choice of topology) is a generator, but it
  is not strong. In fact, $\mathbf{Top}$ does not have a strong generating set
  at all. In the category $\Ban$, the complex numbers $\mathbb C$ is a strong
  generator, whereas in $\OS$, the operator space $\mathbb C$ is a generator,
  but it is not strong (see Section \ref{sec:os}).
\end{example}

\begin{definition}[{\cite[Definition 5.2.1]{borceux2}}]
  \label{def:locally-presentable}
  Let $\alpha$ be a regular cardinal. A category $\CC$ is locally
  \emph{$\alpha$-presentable} if $\CC$ is cocomplete and it has a strong
  generating set consisting of $\alpha$-presentable objects. When $\alpha =
  \aleph_0$, we say that $\CC$ is locally \emph{finitely} presentable and when
  $\alpha = \aleph_1$, we say that $\CC$ is locally \emph{countably}
  presentable. A category $\CC$ is \emph{locally presentable} if it is locally
  $\alpha$-presentable for some regular cardinal $\alpha.$
\end{definition}

\section{The Category of Operator Spaces}
\label{sec:os}


We begin by recalling some background on operator spaces and the
Heisenberg-\Schrod{} duality. Then we establish results related to the
categorical structure of $\OS$.
Our companion paper \cite{OS-cat} provides
detailed proofs of the more difficult and technical results (e.g. local presentability)
and provides relevant references to the literature that we used to prove many of the
other (straightforward) results.

\subsection{Preliminaries on Operator Spaces}
\label{sub:os-preliminaries}

We begin by recalling some background on operator spaces and we also use this
as an opportunity to fix notation. Most of the material we present in this
subsection is standard and it is based on the textbook accounts
\cite{effros-ruan,blecher-merdy,pisier}.

\begin{definition}[Matrix Space]
  Let $V$ be a vector space. We write $\MM_n(V)$ for the vector space
  consisting of the $n \times n$ matrices with matrix entries in the vector
  space $V$ (with vector space operations defined componentwise).  When $V =
  \mathbb C$, we often write $\mathbb M_n$ for $\mathbb M_n(\mathbb C).$
\end{definition}

The vector space $\MM_n$ can be equipped with a Banach space norm via the linear isomorphism
$\MM_n \cong B(\mathbb C^n)$ with the space of (bounded) linear
operators on the Hilbert space $\mathbb C^n$. The space $B(\mathbb C^n)$ has a
canonical norm (i.e.  the operator norm) which can then be used to define a
norm on $\MM_n$ via the above isomorphism. Writing $M_n$ for the
corresponding Banach space, the aforementioned linear isomorphism now becomes
an isometric isomorphism $M_n \cong B(\mathbb C^n).$

\begin{definition}
  \label{def:operator-space}
  An \emph{(abstract) operator space} is a complex vector space $X$ together with a family of norms
  \[ \{ \norm{\cdot}_n \colon \MM_n(X) \to [0, \infty) \ |\ n \in \mathbb N \} , \]
  such that:
  \begin{enumerate}
    \item[(B)] The pair $(\MM_1(X), \norm{\cdot}_1)$ is a Banach space;
    \item[(M1)] $\|x\oplus y\|_{m+n}=\max\{\|x\|_m,\|y\|_n\}$
    \item[(M2)] $\|\alpha x\beta\|_m\leq\|\alpha\|\|x\|_m\|\beta\|$
  \end{enumerate}
  for each $n,m\in\mathbb N$, $x\in\MM_m(X)$, $y\in\MM_n(X)$, $\alpha,\beta\in M_{m}$.
  Here $x\oplus y\in \MM_{m+n}(X)$ is defined as the matrix 
  \[ 
    x\oplus y\eqdef
    \begin{bmatrix}
      x & 0\\
      0 & y
    \end{bmatrix}
  \]
  and $\alpha x \beta$ is defined through the obvious generalisation of matrix
  multiplication.
  For an operator space $X$, we write $M_n(X)$ for the normed space $(\MM_n(X),
  \norm{\cdot}_n)$. We also write $\norm{\cdot} \colon X \to [0, \infty)$ for the norm
  defined on $X$ through the linear isomorphism $M_1(X) \cong X.$ We
  call a family of norms that satisfies the above criteria an \emph{operator
  space structure} (OSS) on the vector space $X$.
\end{definition}

Given an operator space $X$, it is clear that $X$ is a Banach space with
respect to the norm $\norm{\cdot}$ above and textbook results show that
each space $M_n(X)$ is also a Banach space. Therefore, to give an OSS on $X$
is to give a sequence of Banach spaces $M_n(X)$ that satisfy (M1) and (M2).
We often describe operator space structures in this manner.
Next, we introduce some of the types of morphisms of operator spaces.

\begin{definition}
  \label{def:os-maps}
  Let $u \colon X\to Y$ be a linear map between vector spaces $X$ and $Y$. We write
  $u_n:\MM_n(X)\to\MM_n(Y)$ for the linear map  $[x_{ij}]\mapsto [u(x_{ij})]$,
  i.e. the map defined by component-wise application of $u$. If
  $X$ and $Y$ are operator spaces, we say that $u$ is \emph{completely bounded} if
  $\|u\|_{\mathrm{cb}}\eqdef\sup_{n\in\mathbb N}\|u_n\|<\infty$, where $\norm{u_n}$
  is the operator norm of the map $u_n \colon M_n(X) \to M_n(Y)$ between
  the indicated Banach spaces.
  We write $\CB(X,Y)$ for the Banach
  space of all completely bounded maps from $X$ to $Y$ equipped with the
  $\| \cdot \|_{\mathrm{cb}}$-norm. Furthermore, we say that a linear map $u \colon X\to Y$
  between operator spaces is:
  \begin{itemize}
    \item a \emph{complete contraction}, if $u_n$ is a contraction for each $n \in \mathbb N$, equivalently if $\norm{u}_{\mathrm{cb}} \leq 1$;
    \item a \emph{complete isometry}, if $u_n$ is an isometry for each $n \in \mathbb N$;
    \item a \emph{complete quotient map} (also known as \emph{complete metric surjection})
      if $u_n$ is a quotient map for each $n \in \mathbb N$, i.e. if every map $u_n \colon M_n(X) \to M_n(Y)$ maps the open unit ball of $M_n(X)$ \emph{onto} the open unit ball of $M_n(Y)$.
    \item a \emph{completely isometric isomorphism}, if $u$ is a surjective complete isometry.
  \end{itemize}
\end{definition}

\begin{remark}
  The map $u_n \colon M_n(X) \to M_n(Y)$ is called the $n$-th
  \emph{amplification} of $u$ and may be thought of as the map $\id \otimes u
  \colon M_n \otimes X \to M_n \otimes Y.$ See the \emph{injective} (also known
  as \emph{minimal}) tensor product of operator spaces \cite[Chapter 8]{effros-ruan}.
  Amplifications are important in quantum theory because they model the behaviour of operations
  when adding auxiliary systems. Operator spaces allow us to talk about the continuity, boundedness,
  contractivity, etc. of the amplifications $u_n$.
\end{remark}

\begin{remark}
  In order to avoid repetition, all maps between vector spaces, Banach spaces,
  and operator spaces are implicitly assumed to be linear unless stated
  otherwise.
\end{remark}

Let $V$ be a vector space and let $X$ be a Banach space for which there is a
linear isomorphism $\varphi \colon V\to X$. Then $V$ can be equipped with a norm
$\| \cdot \|$ defined by $\|v\|\eqdef\|\varphi(v)\|$, and when equipped with this
norm, $V$ becomes a Banach space and $\varphi$ lifts to an isometric isomorphism $\varphi \colon V \to X.$
If, moreover, $X$ is an operator space, then $V$ can be equipped
with an OSS by defining for each $n\in\mathbb N$ the norm $\|\cdot \|_n$ on
$\MM_n(V)$ by $\|[x_{ij}]\|_n\eqdef\|[\varphi(x_{ij})]\|_n$ for each
$[x_{ij}]\in \MM_n(V)$. With this OSS, $\varphi \colon V \to X$ lifts to a
completely isometric isomorphism.
Similarly, if $\varphi \colon V \to X$ is an
isometry from a Banach space $V$ to an operator space $X$, the same
construction allows us to equip $V$ with an OSS such that $\varphi$ lifts to a
complete isometry. This gives a convenient way to define operator space
structures that are used in the literature.

\begin{example}
  The vector space of complex numbers $\mathbb C$ has a unique operator space
  structure \cite[Chapter 3]{pisier} given by $M_n(\mathbb C) \eqdef M_n$ (see previous example).
\end{example}

\begin{example}
  For an operator space $X$, each of the Banach spaces $M_n(X)$ have a
  canonical operator space structure as well. In particular, we can define a
  norm on $\MM_m(M_n(X))$ through the linear isomorphism
  $\MM_m(M_n(X)) \cong M_{mn}(X)$ and these norms determine an OSS on $M_n(X).$
  This gives the canonical OSS on $M_n = M_n(\mathbb C)$.
\end{example}

\begin{example}
  \label{ex:bounded-hilb}
  Given two Hilbert spaces $H$ and $K$, the space $B(H,K)$ of all bounded
  operators from $H$ to $K$ has a canonical OSS that is defined through the
  linear isomorphism $\MM_{n}(B(H,K))\cong B(H^{\oplus n},K^{\oplus n})$ and
  the Banach space structure of the latter space, where $H^{\oplus n}$ is the
  $n$-fold Hilbert space direct sum of $H$ with itself \cite[1.2.2]{blecher-merdy}. When $H = K$, we simply
  write $B(H)$ for the corresponding operator space. If $H = \mathbb C^n$,
  then $B(H) \cong M_n$ completely isometrically. The matrix space
  $\mathbb M_{m,n}$ of $m \times n$ complex matrices can be equipped with an OSS through the linear isomorphism
  $\mathbb M_{m,n} \cong B(\mathbb C^{\oplus n}, \mathbb C^{\oplus m})$ and the OSS of the latter space,
  and we write $M_{m,n}$ to refer to it.
\end{example}

\begin{example}\label{ex:C-star-algebra-OSS}
  Every von Neumann algebra (e.g. $B(H)$), and more generally, every C*-algebra has a
  canonical OSS. If $A$ is a C*-algebra (von Neumann algebra), then each matrix
  space $\MM_n(A)$ has a unique norm under which it can be equipped with the
  structure of a C*-algebra (von Neumann algebra). These norms then give $A$
  its canonical OSS. In the special case where we consider the von Neumann
  algebra $B(H)$, this is consistent with the OSS from Example
  \ref{ex:bounded-hilb}.
\end{example}

\begin{example}\cite[Section 3.2]{effros-ruan}
    Let $X$ and $Y$ be operator spaces. Then $\CB(X,Y)$ has a canonical
    operator space structure that can be defined via the linear isomorphism $
    \MM_n( \CB(X,Y) )\cong \CB(X,M_n(Y))  $ and the Banach space structure of
    the latter space.
\end{example}

A special case of the previous example is obtained if $Y=\mathbb C$. Since for
any commutative C*-algebra $A$ and any bounded linear map $\varphi:X\to A$ we
have $\|\varphi\|_{\mathrm{cb}}=\|\varphi\|$ \cite[Proposition 2.2.6]{effros-ruan}, it
follows that $\CB(X,\mathbb C)=B(X,\mathbb C)=X^*$, where $X^*$ denotes the
Banach space dual of $X$. Hence we have a linear isomorphism $\MM_n(X^*)\cong
\CB(X,M_n)$ \cite[p. 41]{effros-ruan}, which defines a norm on $\MM_n(X^*)$.
Denoting the corresponding normed space by $M_n(X^*)$, we obtain an isometric
isomorphism $M_n(X^*) \cong \CB(X,M_n)$ which yields an operator space structure
on $X^*$.

\begin{example}
  Let $H$ be a Hilbert space and let $T(H) \subseteq B(H)$ be the subset
  consisting of the \emph{trace-class operators}. We omit the details here, but
  what is important to know is that we have a bounded linear functional $\trace \colon T(H) \to \mathbb C$ 
  which gives the trace of an operator. The space $T(H)$ becomes a Banach space when equipped with the
  trace norm and then we obtain the well-known Banach space duality $T(H)^* \cong B(H)$.
  Furthermore, $T(H)$ has a canonical OSS which is inherited from the
  isometry $T(H) \to T(H)^{**} \cong B(H)^*$ and the OSS of the latter space.
  Equipped with this OSS, we have a completely isometric isomorphism $T(H)^* \cong B(H)$
  \cite[Theorem 3.2.3]{effros-ruan}.
\end{example}

\begin{example}
  The matrix space $\MM_n(\mathbb C)$ can be equipped with another OSS that is
  called the \emph{trace class} OSS. This can be achieved through the linear
  isomorphism $\MM_n(\mathbb C) \cong T(\mathbb C^n)$ and the OSS of the
  latter space. We write $T_n$ for the corresponding operator space.
\end{example}

\begin{remark}
  Thanks to a fundamental result due to Ruan (see \cite[Theorem
  2.3.5]{effros-ruan}), for every abstract operator space $X,$ we can construct
  a Hilbert space $H$ such that there exists a completely isometric embedding
  $X \hookrightarrow B(H)$. Conversely, a closed linear subspace $X \subseteq
  B(H)$ inherits an OSS from $B(H)$ such that the inclusion is completely
  isometric, and we refer to $X$ as a \emph{concrete} operator space. Ruan's
  result establishes the equivalence between the two notions and
  justifies using the term ``operator space''.
\end{remark}

\subsection{Preliminaries on the Heisenberg-\Schrod{} Duality}
\label{sub:hs-duality}

As we have already pointed out, both $B(H)$ and $T(H)$ have canonical operator
space structures. Furthermore, we have $B(H) \cong T(H)^*$ completely
isometrically, and we say that $B(H)$ is the dual space of $T(H)$. We also
write $T(H) \cong B(H)_*$ to indicate that $T(H)$ is completely isometrically
isomorphic to the \emph{predual} of $B(H).$
The subset $T(H)$ forms a two-sided ideal of $B(H)$ and the bilinear form
\begin{align}
  \trace \colon T(H) \times B(H) &\to \mathbb C \label{eq:hs-type} \\
  \trace(x,b) &= \trace(xb) \label{eq:hs-def}
\end{align}
provides us with the structure of a \emph{dual pairing}. In the above expression
$\trace(xb)$ stands for the trace of the operator $xb \in T(H)$.
If $\psi \colon T(H_1) \to T(H_2)$ is a bounded linear map, then there exists a \emph{unique} bounded linear map $\psi^t \colon B(H_2) \to B(H_1)$, called the transpose of $\psi$, such that
\begin{equation}
  \label{eq:hs-dual-system}
  \trace(\psi(x), b) = \trace(x, \psi^t(b)) 
\end{equation}
for every $x \in T(H_1)$ and $b \in B(H_2)$. The transpose $\psi^t$ is determined by
the Banach space adjoint $\psi^*$ and if we trivialise the isomorphism $B(H) \cong
T(H)^*$, then the two coincide. Therefore, intuitively, one may think of the
transpose $\psi^t$ and the dual map $\psi^*$ as being essentially the same thing
(modulo isomorphism). The dual pairing \eqref{eq:hs-type} can be used to define
the Heisenberg-\Schrod{} duality and we proceed to formulate it in greater
detail.

Recall that a linear map $\varphi \colon B(H_1) \to B(H_2)$ is called:
\begin{itemize}
  \item \emph{positive}, if it preserves positive elements;
  \item \emph{completely-positive}, if the map $\varphi_n \colon M_n(B(H_1)) \to M_n(B(H_2))$ is positive for every $n \in \mathbb N$
    with respect to the canonical von Neumann algebra structures of $M_n(B(H_i))$;
  \item \emph{unital}, if $\varphi(1_{H_1}) = 1_{H_2}$, i.e. it preserves the unit element;
  \item \emph{normal}, if there exists a (necessarily unique) bounded linear map $\psi \colon T(H_2) \to T(H_1),$ such that
    $\varphi = \psi^t$, i.e. $\varphi$ is the transpose of another bounded linear map\footnote{Equivalently, maps that are ultraweakly continuous.}.
\end{itemize}
It follows that every (completely-)positive map $\varphi$ is bounded and every normal linear map is also bounded.

Turning to the \Schrod{} picture, a bounded linear map $\psi \colon T(H_1) \to T(H_2)$ is called
\emph{trace-preserving} if $\trace(\psi (x)) = \trace(x)$ for every $x \in T(H_1)$
and we may similarly define what it means for $\psi$ to be
(completely-)positive. This allows us to formulate the following correspondence
between maps in the Heisenberg and \Schrod{} pictures of quantum theory.

\begin{restatable}{proposition}{hsCorrespond}
  \label{prop:hs-correspondence}
  Let $H_1$ and $H_2$ be Hilbert spaces. The transpose operation $(-)^t$ gives
  a bijective correspondence between bounded linear maps $\psi \colon T(H_1) \to T(H_2)$
  and normal linear maps $\psi^t \colon B(H_2) \to B(H_1)$. Furthermore:
\begin{enumerate}
  \item $\psi$ is (completely-)positive iff $\psi^t$ is (completely)-positive;
  \item $\psi$ is trace-preserving iff $\psi^t$ is unital;
  \item $\psi$ is a (complete) contraction iff $\psi^t$ is a (complete) contraction.
\end{enumerate}
\end{restatable}
\begin{proof}
  These facts are already known, see Appendix \ref{app:hs} for more information
  and references.
\end{proof}

It follows that there exists a bijective correspondence
\begin{equation}
  \label{eq:hs-duality}
  \CPTP(T(H_1), T(H_2)) \cong \NCPU(B(H_2), B(H_1))
\end{equation}
between completely-positive trace-preserving (CPTP) maps and normal
completely-positive unital (NCPU) maps. The \emph{quantum operations} (also
known as \emph{quantum channels}) in the \Schrod{} picture correspond to CPTP
maps $T(H_1) \to T(H_2)$, whereas the quantum operations in the Heisenberg
picture correspond to NCPU maps $B(H_2) \to B(H_1)$. The
correspondence \eqref{eq:hs-duality} is known as the \emph{Heisenberg-\Schrod{}
duality}.

\subsection{Limits and Colimits in $\OS$}
\label{sub:os-colimits}

Now we show how (co)products and (co)equalisers are constructed
in the category of operator spaces.

\begin{definition}
  We write $\OS$ for the (locally small) category whose objects are the
  operator spaces and whose morphisms are the linear complete contractions
  between them.
\end{definition}

It is easy to see that $f \colon X \to Y$ in $\OS$ is an isomorphism in the
categorical sense iff $f$ is a completely isometric isomorphism of operator
spaces. The category $\OS$ has a zero object, written $0$, which is given by
the zero-dimensional operator space. We write $0_{X,Y} \colon X \to Y$ for the
zero maps, which we sometimes abbreviate by writing $0$ as well.

Recall that, given an indexed family $(X_\lambda)_{\lambda\in\Lambda}$ of Banach 
spaces, their $\ell^\infty$-direct sum is the Banach space
\[
  \bigoplus^\infty_{\lambda\in\Lambda}X_\lambda \eqdef \left\{ 
(x_\lambda)_{\lambda \in \Lambda} \in \prod_{\lambda\in\Lambda}X_\lambda   \ |\ 
  \sup_{\lambda\in\Lambda}\norm{x_\lambda} < \infty \right\}
\]
equipped with the $\ell^\infty$-norm, i.e. $\norm{(x_\lambda)_{\lambda\in\Lambda}}_\infty\eqdef\sup_{\lambda\in\Lambda}\norm{x_\lambda}$.

\begin{proposition}[\cite{OS-cat}]
  \label{prop:products}
  Let $(X_\lambda)_{\lambda\in\Lambda}$ be an indexed family of operator spaces. Then $\bigoplus^\infty_{\lambda\in\Lambda}X_\lambda$ can be equipped with an OSS such that 
  \begin{itemize}
    \item[(a)] $M_n\left(\bigoplus_{\lambda\in\Lambda}^\infty X_\lambda\right)=\bigoplus^\infty_{\lambda\in \Lambda}M_n(X_\lambda)$.
    \item[(b)] for each $\kappa\in\Lambda$, the canonical inclusion $\iota_\kappa:X_\kappa\to \bigoplus^\infty_{\lambda\in\Lambda}X_\lambda$ is a complete isometry;
    \item[(c)] for each $\kappa\in \Lambda$, the canonical projection $\pi_\kappa \colon \bigoplus^\infty_{\lambda\in \Lambda}X_\lambda \to X_\kappa$ is a complete quotient map;
    \item[(d)] the operator space $\bigoplus^\infty_{\lambda\in\Lambda}X_\lambda$ together with the projections $\pi_\lambda$ constitute the categorical product of the family $(X_\lambda)_{\lambda \in \Lambda}.$
  \end{itemize}
\end{proposition}

For the coproduct construction, we first need to recall the definition of
potentially uncountable sums in a Banach space.
Let $\Lambda$ be a (potentially uncountable) index set. We say that the sum of an
indexed family $(x_\lambda)_{\lambda\in\Lambda}$ of vectors in a Banach space $X$
\emph{converges} if there is $x\in X$ such that the net $\left\{\sum_{\lambda\in F}x_\lambda:F\subseteq\Lambda\text{ finite}\right\}$ converges to $x$,
in which case we write $\sum_{\lambda\in\Lambda}x_\lambda=x$.
The convergence of the net to $x$ implies that for each
$\epsilon>0$ there is a finite subset $F\subseteq \Lambda$ such that for each
finite subset $F \subseteq G\subseteq \Lambda,$ we have
$\left\|\sum_{\lambda \in G} x_\lambda-x\right \| < \epsilon$. We also recall that a sufficient
condition for $\sum_{\lambda \in \Lambda} x_\lambda$ to exist is that $\sum_{\lambda \in \Lambda} \norm{ x_\lambda } < \infty$ and in this case
$\norm{\sum_{\lambda \in \Lambda} x_\lambda} \leq \sum_{\lambda \in \Lambda} \norm{ x_\lambda } . $

Recall that, given an indexed family $(X_\lambda)_{\lambda\in\Lambda}$ of Banach spaces,
their $\ell^1$-direct sum is the Banach space
\[ 
    \bigoplus_{\lambda\in\Lambda}^1X_\lambda \eqdef \left\{ (x_\lambda)_{\lambda\in\Lambda} \in \prod_{\lambda\in\Lambda}X_\lambda \ |\ \sum_{\lambda\in\Lambda}\norm{x_\lambda}<\infty  \right\}
\]
equipped with the $\ell^1$-norm, i.e.
$\norm{(x_\lambda)_{\lambda\in\Lambda}}_1\eqdef\sum_{\lambda\in\Lambda}\norm{x_\lambda}$.
Note that $x_\lambda = 0$ for all but countably many indices $\lambda$.

\begin{proposition}[\cite{OS-cat}]
  \label{prop:coproducts}
  Let $(X_\lambda)_{\lambda\in\Lambda}$ be an indexed family of operator spaces. Then $\bigoplus_{\lambda\in\Lambda}^1X_\lambda$ can be equipped with an OSS such that 
  \begin{itemize}
    \item[(a)] there exists a complete isometric isomorphism of dual operator spaces $\left(\bigoplus^1_{\lambda\in\Lambda}X_\lambda\right)^*\cong\bigoplus_{\lambda\in\Lambda}^\infty X_\lambda^*$;
    \item[(b)] for each $\kappa\in\Lambda$, the canonical inclusion $\iota_\kappa:X_\kappa\to \bigoplus^1_{\lambda\in\Lambda}X_\lambda$ is a complete isometry;
    \item[(c)] for each $\kappa\in \Lambda$, the canonical projection $\pi_\kappa:\bigoplus^1_{\lambda\in \Lambda}X_\lambda \to X_\kappa$ is a complete quotient map;
    \item[(d)] for each operator space $Y$, there is a complete isometric isomorphism $\CB\left(\bigoplus_{\lambda\in\Lambda}^1X_\lambda,Y\right)\cong\bigoplus^\infty_{\lambda\in\Lambda}\CB(X_\lambda,Y).$
    \item[(e)] the operator space $\bigoplus^1_{\lambda\in\Lambda}X_\lambda$ together with the inclusions $\iota_\lambda$ constitute the categorical coproduct of the family $(X_\lambda)_{\lambda \in \Lambda}.$
  \end{itemize}
\end{proposition}

Next, we describe the construction of equalisers which is completely analogous to the case for Banach spaces.

\begin{proposition}[\cite{OS-cat}]
  \label{prop:equalisers}
  Let $f, g \colon X \to Y$ be two complete contractions. Their equaliser is
  given by the closed subspace $E \eqdef \{ x \in X\ |\ f(x) = g(x) \}$
  together with the completely isometric inclusion $E \subseteq X.$
\end{proposition}

In order to construct coequalisers, we first need the concept of a
quotient operator space. We recall that if $N$ is a closed subspace of a
Banach space, then $X/N$ becomes a Banach space if we define
$\norm{[x]}_{X/N}=\inf_{y\in N}\norm{x+y}=\inf_{y\in[x]}\norm{y}$ for each
$[x]\in X/N$.  Hence, the quotient map $q:X\to X/N$ is always a contraction. If
$X$ is an operator space, then $M_n(N)$ is a closed subspace of $M_n(X)$, hence
$N$ is an operator space, which leads to the next proposition.

\begin{lemma}[{\cite[Proposition 3.1.1]{effros-ruan}}]
  Let $X$ be an operator space and $N\subseteq X$ a closed subspace. Then
  there is an OSS on $X/N$ such that $M_n(X/N)=M_n(X)/M_n(N)$. Moreover, the
  quotient map $q:X\to X/N$ is a \emph{complete quotient map}.
\end{lemma}

\begin{proposition}[\cite{OS-cat}]
  \label{prop:coequalisers}
  Let $f,g \colon X \to Y$ be two complete contractions. Then their coequaliser
  is given by the operator space   $E \eqdef Y / \overline{\Image(f-g)}$
  together with the complete quotient map $q \colon Y \to E :: y \mapsto [y].$
\end{proposition}

\begin{proposition}
  \label{prop:cocomplete}
  The category $\OS$ is complete and cocomplete.
\end{proposition}
\begin{proof}
  Combine the propositions in this subsection.
\end{proof}

\subsection{Strong generators in $\OS$}
\label{sub:generators}

For the locally presentable structure of a category, it is useful to identify
its strong generators. This is the main purpose of this subsection. As a first
step, it is useful to classify the monomorphisms and the epimorphisms in $\OS$.
This is again analogous to $\Ban$ (see \cite{borceux1} for the Banach space case).

\begin{proposition}[{\cite{OS-cat}}]
  \label{prop:monomorphisms}
  Let $f \colon X \to Y$ be a complete contraction between operator spaces $X$
  and $Y.$ Then: 
  \begin{itemize}
    \item $f$ is a monomorphism iff $f$ is an injection;
    \item $f$ is an epimorphism iff the image $f[X]$ is dense in $Y.$
  \end{itemize}
\end{proposition}

Next, we classify special kinds of epimorphisms and monomorphisms that behave
better in locally presentable categories. The situation is again completely
analogous to $\Ban.$

\begin{proposition}[{\cite{OS-cat}}]
  Let $f \colon X \to Y$ be a complete contraction in $\OS$. The following are equivalent:
  \begin{enumerate}
    \item[(1)] $f$ is a regular monomorphism (epimorphism);
    \item[(2)] $f$ is a strong monomorphism (epimorphism);
    \item[(3)] $f$ is an extremal monomorphism (epimorphism);
    \item[(4)] $f$ is a complete isometry (complete quotient map).
  \end{enumerate}
\end{proposition}

\begin{proposition}[{\cite{OS-cat}}]
  \label{prop:strong-generator}
  Each of the following items is a strong generator for $\OS$:
  \begin{enumerate}
    \item[(1)] the set $\{ T_n \ |\ n \in \mathbb N \},$ consisting of the finite-dimensional trace class operator spaces;
    \item[(2)] the operator space $\oplus^1_{n \in \mathbb N} T_n  $;
    \item[(3)] the operator space $T(\ell_2)$, where $\ell_2 \eqdef \ell^2(\mathbb N)$.
  \end{enumerate}
\end{proposition}

It is useful to
understand why the situation is different compared to $\Ban$, where the complex
numbers $\C$ are a strong generator. The reason is similar to the
reason why the terminal object $1$ is a non-strong generator in $\Top.$
We have adjunctions
  \cstikz{triple-adjunction-top.tikz}
where $D$ is the functor that assigns the discrete topology to a set and $I$ is the functor that assigns the indiscrete one.
It is obvious that any map
$f \colon 1 \to I(\mathbb N)$ can also be seen as a map $f \colon 1 \to D(\mathbb N)$, so it factorises through
$\id \colon D(\mathbb N) \to I(\mathbb N)$, which is a proper monomorphism in
$\Top.$ Therefore $1$ is not a strong generator in $\Top$. The proof in $\OS$
is similar.

Let $U \colon \OS \to \Ban$ be the obvious forgetful functor.
Let $\Max \colon \Ban \to \OS$ be the maximal quantisation functor, i.e. the
functor that assigns to a Banach space $X$ the biggest possible OSS in the
sense that $(U \circ \Max)(X) = X$ and such that for each $n \geq 2,$
the norm $\norm{\cdot}_n$ on $M_n(X)$ is the largest norm among all such OSS norms (see
\cite[\S 3.3]{effros-ruan} for more information). We write $\Min \colon \Ban
\to \OS$ for the minimal quantisation functor, which has the property that
$(U \circ \Min)(X) = X$ and which assigns the smallest possible OSS
that is compatible with $X$ (again, see \cite[\S 3.3]{effros-ruan}).
It has already been observed by Yemon Choi that this gives us an
adjoint situation similar to the one in $\Top$ above
\cite{choi-adjunctions}, 
\begin{equation}
  \label{eq:triple-adjunctions}
  \stikz{triple-adjunction-os.tikz}
\end{equation}
see also \cite{pestov-adjunction} where a similar adjunction is used.

\begin{proposition}
  The operator space $\mathbb C$ is a non-strong generator for $\OS$.
\end{proposition}
\begin{proof}
  It is easy to see that $\C$ is a generator (see \cite{OS-cat}). Consider the following diagram in $\OS,$
  \cstikz{non-strong-os.tikz}
  where $X = B(\ell_2)$ viewed as a Banach space.
  The identity map $\id \colon \Max(X) \to \Min(X)$ is a proper monomorphism.
  The diagram above is well-defined in $\OS$ and it commutes for any choice of
  $f \colon \C \to \Min(X)$.
\end{proof}

\subsection{Local Presentability of $\OS$}  
\label{sub:os-lp}

The only finitely-presentable object in $\OS$ (modulo isomorphism) is $0$, i.e.
the zero-dimensional operator space. The argument is straightforward (see
\cite{OS-cat}) and completely analogous to the situation in $\Ban$.
Just as in $\Ban$, the countably-presentable objects in $\OS$ coincide
with the \emph{separable} operator spaces. We recall that an operator space $X$ is separable whenever $X$ is separable
as a Banach space, i.e. when $X$ contains a countable dense subset.
\begin{restatable}{theorem}{osPresentableObject}
  \label{thm:presentable-objects}
  An operator space $X$ is countably-presentable in $\OS$ iff $X$ is separable.
\end{restatable}
\begin{proof}
  The proof requires considerable technical effort. It is included in our
  companion paper \cite{OS-cat}.
\end{proof}

\begin{theorem}
  \label{thm:os-locally-presentable}
  The category $\OS$ is locally countably presentable.
\end{theorem}
\begin{proof}
  The category $\OS$ is cocomplete (Proposition \ref{prop:cocomplete}) and it has a
  strong generator (Proposition \ref{prop:strong-generator}) consisting of
  countably-presentable objects (Theorem \ref{thm:presentable-objects}).
\end{proof}

Our development shows that the locally presentable structures of $\OS$ and
$\Ban$ are closely related to each other. In the theory of locally
$\alpha$-presentable categories, there is an important role that is played by
reflective subcategories that are closed under $\alpha$-directed colimits
\cite{lp-categories-book}. In fact, the relationship between $\Ban$ and $\OS$
satisfies this. To recognise this, recall that a \emph{full} subcategory $\CC$
of $\DD$ is reflective whenever its inclusion $\CC \hookrightarrow \DD$ has a
left adjoint. The functor $\Min$ in \eqref{eq:triple-adjunctions} is a fully
faithful right adjoint between locally countably presentable categories and it
preserves countably-directed colimits
(since $U$ preserves countably-presentable objects, which are exactly the
separable ones, and then by \cite[A.1]{henry-lambda-adjunction}).
Let
$\mathbf{Min}$ be the full subcategory of $\OS$ consisting of minimal operator
spaces, equivalently, operator spaces $X$ such that $X = \Min(U(X)).$

\begin{corollary}
  \label{cor:ban-os}
  We have an isomorphism of categories
  \[ \Ban \cong \mathbf{Min} , \]
  where $\mathbf{Min}$ is a reflective subcategory of $\OS$ closed under
  countably-directed colimits.
\end{corollary}

\subsection{Symmetric Monoidal Closure of $\OS$}
\label{sub:smcc}

Recall that $\Ban$ is a symmetric monoidal closed category with respect to the
projective tensor product $X \pbantimes Y$ of Banach spaces and with
internal hom given by the space of bounded linear maps $B(X,Y)$ between two
Banach spaces. The analogous result for $\OS$ is obtained by replacing the
projective tensor with the \emph{completely} projective tensor product of
operator spaces (see \cite[Chapter 7]{effros-ruan} and \cite[Section
1.5]{blecher-merdy}) and replacing $B(X,Y)$ with the operator space of
\emph{completely} bounded maps $\CB(X,Y).$

\begin{remark}
  What we call the ``completely projective'' tensor is usually referred to
  simply as the ``projective'' tensor in the literature on operator spaces.
  However, in order to avoid confusion with the Banach space projective tensor,
  we choose to use the term ``completely projective'' tensor in this paper.
  The two tensors are different in general.
\end{remark}

Recall that, if $X$ and $Y$ are vector spaces over $\mathbb C$, then the
algebraic tensor product $X\otimes Y$ of $X$ and $Y$ enjoys the following
universal property: there exists a bilinear map $\pi\colon X\times Y\to
X\otimes Y$ such that for each vector space $Z$ and each bilinear map $u\colon
X\times Y\to Z$ there is a \emph{unique} linear map $\tilde u\colon X\otimes
Y\to Z$ such that $\tilde u\circ\pi= u$. We call $\tilde u$ the
\emph{linearization} of $u$. In particular, the linearization of $\pi$ itself
is the identity on $X\otimes Y$. We call elements in the image of $\pi$
\emph{elementary tensors}. In particular, for $x\in X$ and $y\in Y$, we write
$x\otimes y\eqdef \pi(x,y)$.

We can now recall the completely projective tensor product $\ptimes$ of
operator spaces and the universal property that it enjoys.

\begin{definition}
  \label{def:projective-tensor-operator}
  Let $X$ and $Y$ be operator spaces. For an element $v \in \MM_n(X \otimes Y)$, consider the norm
  \begin{align*}
    \norm{v}_{\wedge} \eqdef \inf \{ & \norm \alpha \norm x \norm y \norm \beta \  |\ \\
    & p \in \mathbb N, q \in \mathbb N, x \in M_p(X) , y \in M_q(Y), \\
    & \alpha \in M_{n, pq}, \beta \in M_{pq, n} , \text{ and } v = \alpha (x \otimes y) \beta\} .
  \end{align*}
  Here $x\otimes y\in M_{pq}(X\otimes Y)$ is the ``tensor product of matrices'' defined by
  \[x\otimes y\eqdef [x_{ij}\otimes y_{kl}]_{(i,k),(j,l)}.\]
  Note that the expression for the norm immediately yields $\|x\otimes y\|_{\wedge}\leq \|x\|\|y\|$ for such $x\otimes y\in \MM_{pq}(X\otimes Y)$. 
  Less obviously, it is also true that $\|x\otimes y\|_{\wedge} = \|x\|\|y\|.$
  We write $X \ptimes Y$ for the completion of $X \otimes Y$ with respect to the above norm (on $\MM_1(X \otimes Y)$) and say that $X \ptimes Y$ is the 
  completely projective tensor product
  of $X$ and $Y$. More specifically, $M_n(X \ptimes Y)$ is given by the completion of $\MM_n(X \otimes Y)$ with respect to the above norm and this determines the OSS.
\end{definition}

A proof that $X\ptimes Y$ is an operator space can be found in
\cite[1.5.11]{blecher-merdy} or \cite[Section 7.1]{effros-ruan}. For the
universal property of this tensor, we recall the appropriate kinds of bilinear
maps.

\begin{definition}
  Let $X$, $Y$ and $Z$ be operator spaces. Then a bilinear map $u:X\times Y\to Z$ is called \emph{jointly completely bounded} if there exists a $K\geq 0$ such that for each $n,m\in\mathbb N$ and each $[x_{ij}]\in M_n(X)$ and each $[y_{kl}]\in M_m(Y)$, we have 
    \[ \|[u(x_{ij},y_{kl})]_{(i,k),(j,l)}\|\leq K\| [x_{ij}]\|\|[y_{kl}]\|.\]
    Note that $[u(x_{ij},y_{kl})]_{(i,k),(j,l)}\in M_{nm}(Z)$, so it is an $nm\times nm$ matrix of elements in $Z$.
  We define $\|u\|_{\mathrm{jcb}}$ to be the least $K$ satisfying the above condition.
  If $\|u\|_{\mathrm{jcb}}\leq 1$, we say that $u$ is \emph{jointly completely contractive}. 
  We write $\JCB(X\times Y, Z)$ for the space of all jointly completely bounded bilinear maps $X\times Y\to Z$. It follows that $\|\cdot\|_{\mathrm{jcb}}$ is a norm on $\JCB(X \times Y, Z)$ and it can also be seen as a matrix norm through the linear isomorphism
  \[ \mathbb M_n(\JCB(X \times Y, Z)) \cong \JCB(X \times Y, M_n(Z)) . \]
\end{definition}

Given two operator spaces $X$ and $Y$, the map
$\pi_{X,Y}\colon X\times Y\to X\otimes Y$ extends to a jointly completely
contractive map $X\times Y\to X\ptimes Y$, which we also denote by
$\pi_{X,Y}$, and which satisfies the following universal property.

\begin{proposition}[{\cite[1.5.11]{blecher-merdy}\cite[Proposition 7.1.2]{effros-ruan}}]
  \label{prop:universal-tensor}
  Given operator spaces $X$, $Y$, $Z$, and a jointly completely
  contractive (bounded) map $u\colon X\times Y\to Z,$ there is a unique
  completely contractive (bounded) map $\bar u:X\ptimes Y\to Z$ such that
  \begin{equation}\label{eq:universal property}
    \bar u\circ\pi_{X,Y} =u,
  \end{equation}
  which is obtained as the unique continuous extension of
  $\tilde u:X\otimes Y\to Z$. Moreover, we have
  $\|u\|_{\mathrm{jcb}}=\|\bar u\|_{\mathrm{cb}}$. In particular, we have a completely
  isometric isomorphism
  \begin{align*}
    JCB(X\times Y,Z) &\cong \CB(X\ptimes Y,Z) \\ 
    u &\mapsto \bar u.
  \end{align*}
\end{proposition}

Building on the above proposition, one can show that $\OS$ is a symmetric
monoidal closed category. We believe that this might be folklore knowledge and
the monoidal closure of $\OS$ has been explicitly pointed out by Yemon Choi in
an online discussion already \cite{choi-smcc}. In presenting Proposition
\ref{prop:os-smcc}, our intention is not to claim originality of this result,
but to use it for proving additional results.

\begin{proposition}
  \label{prop:os-smcc}
  The category $\OS$ has the structure of a symmetric monoidal closed category
  with monoidal product $X \ptimes Y$, monoidal unit $\mathbb C$, and
  internal hom $\CB(X,Y)$.
\end{proposition}

\subsection{Linear Logic Exponentials}
\label{sub:exponentials}

There are (at least) two different ways to interpret the
exponential of linear logic in $\OS$. One of them is the Lafont exponential
\cite{lafont-thesis} that we describe next. We write $\CoComon$ for the
category whose objects are the \emph{cocommutative comonoids} internal to $\OS$
(with its symmetric monoidal structure from \secref{sub:smcc}) with
morphisms the \emph{comonoid homomorphisms}.

\begin{theorem}
  There exists an adjunction
  \begin{equation}
    \label{eq:cocommutative-cofree}
    \stikz{free-comonoid-adjunction-os.tikz}
  \end{equation}
  where the \emph{left} adjoint $U$ is the forgetful functor.
\end{theorem}
\begin{proof}
  The category $\OS$ is locally presentable and symmetric monoidal closed, so the proof follows using categorical arguments
  \cite[pp. 10 and pp. 13]{porst-comonoids}. The proof shows that $\CoComon$
  is also locally presentable and uses the special adjoint functor theorem
  to prove that $R$ exists.\footnote{In fact, it even
  follows that $\CoComon$ is cartesian closed \cite{porst-comonoids}.}
\end{proof}

For an operator space $X$, the right adjoint gives the \emph{free}
cocommutative comonoid $RX$ over $X$ (see \cite[\S 7.2]{mellies-linear-logic}
for more information). We write $! \colon \OS \to \OS$ for the induced comonad,
i.e the Lafont exponential.

The other exponential is induced by the adjunction
\begin{equation}
  \label{eq:s-adjunction}
  \stikz{set-os.tikz}
\end{equation}
where $\Ball$ is the unit ball functor, i.e.
$ \Ball(Y) \eqdef \{ y \in Y : \norm{y} \leq 1 \} . $
This adjunction is symmetric monoidal and it is an instance of a copower
adjunction, because we have natural isomorphisms
$ \ell^1(X) \cong \coprod_X \mathbb C$ and $\Ball(Y) \cong \OS(\mathbb C, Y) . $
This is a model of ILL in the sense of Benton \cite{benton-lnl} and we write
$S \eqdef \ell^1 \circ \Ball \colon \OS \to \OS$ for the induced
exponential.

\section{Pure and Mixed Quantum Information in $\OS$}
\label{sec:pure-mixed}

The category $\OS$ allows us to describe both pure and mixed quantum primitives
in a compositional way that is compatible with its structure as a model of ILL.
Furthermore, it also allows us to describe morphisms that model the interaction
between the pure and mixed quantum primitives.

\subsection{Morphisms for Mixed State Quantum Computation}

We begin with a couple of propositions that should clarify the connection
between completely-positive maps and complete contractions.

\begin{proposition}
  \label{prop:cc-cp-vn}
  Let $\varphi \colon B(H_1) \to B(H_2)$ be a linear unital map. Then,
  $\varphi$ is completely-positive iff $\varphi$ is a complete contraction.
\end{proposition}
\begin{proof}
  Special case of \cite[Corollary 5.1.2]{effros-ruan}.
\end{proof}

The above proposition is applicable to quantum operations in the Heisenberg
picture, i.e. NCPU maps. The analogous proposition for the \Schrod{} picture is
given next.

\begin{proposition}
  \label{prop:cc-cp-tp}
  Let $\varphi \colon T(H_1) \to T(H_2)$ be a linear trace-preserving map.
  Then, $\varphi$ is completely-positive iff $\varphi$ is a complete
  contraction.
\end{proposition}
\begin{proof}
  In both cases $\varphi$ is a bounded map. The proof follows by combining
  Proposition \ref{prop:cc-cp-vn} and Proposition \ref{prop:hs-correspondence}.
\end{proof}

This shows that first-order quantum operations, i.e. CPTP/NCPU maps in the
\Schrod{}/Heisenberg picture, are indeed complete contractions and therefore
morphisms in $\OS.$ Indeed, these operations are defined in $\OS$ in the
usual way. We showcase some operations in the \Schrod{} picture
\begin{align*}
  \mathrm{state}_\rho & \colon \mathbb C \to T(H)            :: a \mapsto a \rho \\
  \mathrm{apply}_u    & \colon T(H) \to T(H)                 :: t \mapsto u t u^\dagger \\
  \mathrm{meas}       & \colon T(\ell^2(X)) \to T(\ell^2(X)) :: t \mapsto \sum_{x \in X} \ket x \bra{x} t \ket{x} \bra x
\end{align*}
and the corresponding maps in the Heisenberg picture, when defined in the
standard way, are also morphisms in
$\OS$. The CPTP map $\mathrm{state}_\rho$ prepares a mixed quantum state
determined by the density operator $\rho \in T(H)$, i.e. a positive operator
$\rho \geq 0$ with $\trace(\rho) = 1.$ The CPTP map $\mathrm{apply}_u$
describes the unitary evolution of the system (in the mixed paradigm) with
respect to a unitary operator $u \colon H \to H$. The CPTP map $\mathrm{meas}$
performs a measurement in the obvious basis determined by the set $X$. Other
kinds of measurements are also possible, but for simplicity, we elide them from
our presentation. If $X = \{0, 1\}$, then we
recover the standard measurement in the computational basis of a qubit.
If $X = \mathbb N$, then the measurement outcome may
have infinite support.

\subsection{Interaction between Pure and Mixed Quantum Primitives}

In order to gain some intuition as to why we can describe pure quantum
information within $\OS$, recall that from Example \ref{ex:bounded-hilb}, we
know that if $H$ and $K$ are Hilbert spaces, then $B(H,K)$ has a canonical OSS.
Furthermore, it has been observed that every Hilbert space $H$ may be equipped
with an OSS itself \cite[Section 3.4]{effros-ruan}. In particular, the obvious
isometry
\[ H \cong B(\mathbb C, H) \]
and the OSS of the latter space can be used to equip $H$ with an OSS. We write
$H_c$ for this OSS on $H$ and say that $H_c$ is the \emph{column Hilbert
operator space} determined by $H$ \cite[Section 3.4]{effros-ruan}.
One can then show that we have a completely isometric isomorphism 
\begin{equation}
  \label{eq:hilbc-os-smcc}
  B(H,K) \cong \CB(H_c, K_c) \qquad \text{\cite[Theorem 3.4.1]{effros-ruan}}
\end{equation}
and by restricting to the unit balls we see that
\begin{equation}
  \label{eq:hilbc-os}
  \Hilbc(H,K) \cong \OS(H_c, K_c)
\end{equation}
where $\Hilbc$ is the category whose objects are Hilbert spaces and whose
morphisms are the linear contractions between them. This category is
accessible \cite[pp. 70]{lp-categories-book}, has many other important
categorical properties \cite{hilb-contractions} and the category is clearly
relevant for the study of pure quantum information. From \eqref{eq:hilbc-os}
we see that we may identify $\Hilbc$ with a \emph{full} subcategory of $\OS$.
Note that $\Hilbc$ is not monoidal
closed and from \eqref{eq:hilbc-os-smcc} we see that the internal hom of $\OS$
gives a natural way to ameliorate this. From \eqref{eq:hilbc-os-smcc} we also
see that the higher-order structure of $\OS$ behaves very well with respect to
$\Hilbc$ and with respect to pure quantum primitives.

Keeping the above intuition in mind, let us now explain how the exponential $S$
can be used. Let $H$ and $K$ be Hilbert spaces and consider the following
\emph{non-linear} functions:
\begin{align*}
  & \Uadjoint \colon B(H,K) \to B(K,H) \\
  & \Uadjoint(f) \eqdef f^\dagger \\
  & \Uctrl \colon B(H) \to B(\mathbb C^2 \htimes H) \\
  & \Uctrl( f ) \eqdef (\ket 0 \bra 0 \otimes \id) + (\ket 1 \bra 1 \otimes f) \\
  & \Uapply \colon B(H,K) \to \CB( T(H), T(K)) \\
  & \Uapply(f) \eqdef f(\cdot)f^\dagger
\end{align*}
where we write $\htimes$ for the tensor product of Hilbert spaces.
The action of $\htimes$ on elementary tensors (e.g.
$(\ket 0 \bra 0 \otimes \id)$) is consistent with the algebraic tensor product
$\otimes$, so the notation above is appropriate. The function $\Uadjoint$
returns the Hermitian adjoint of a bounded linear map and $\Uctrl$ returns a
``controlled'' version of a bounded linear map. Both $\Uadjoint$ and $\Uctrl$
should be understood as acting on pure quantum primitives and their use in
pure quantum computation is ubiquitous -- these functions are very often
implicitly used in the description of quantum algorithms and quantum
circuits. The function $\Uapply$ takes a bounded linear map $f$ (pure quantum
primitive) and returns the completely-positive map (mixed quantum primitive)
which performs essentially the same operation as $f$, but in the sense of mixed
state computation.

The (co)domains of these higher-order functions are indeed operator spaces, as
we explained above, but the functions are not linear in the $f$ argument. In
order to deal with the problem of non-linearity, we can use the $S$ exponential
from adjunction \eqref{eq:s-adjunction} and we can show that it allows us to
\emph{linearise} certain kinds of non-linear functions.
From the universal property of \eqref{eq:s-adjunction}
\begin{equation}
  \label{eq:set-os}
\OS(SX, Y) \cong \Set(\Ball(X), \Ball(Y))
\end{equation}
we see that any function
$\underline g \colon \Ball(X) \to \Ball(Y)$ uniquely
determines a linear complete contraction $g \colon SX \to Y$ and vice-versa.
Note that the function $\underline g$ is not required to be linear.
The functions $\Uadjoint, \Uctrl,$ and $\Uapply$ (co)restrict to the unit
balls of the corresponding operator spaces (see Appendix
\ref{app:quantum-primitives}) and this allows us to define linear complete
contractions
\begin{align*}
  & \adjoint \colon S(B(H,K)) \to B(K,H) \\
  & \ctrl \colon S(B(H)) \to B(\mathbb C^2 \htimes H) \\
  & \apply \colon S(B(H,K)) \to \CB( T(H), T(K))
\end{align*}
which are now morphisms in $\OS.$ The adjunction \eqref{eq:set-os}
ensures that for any $x \in X$ with $\norm x \leq 1$, we have
\[ g(\delta_x) = \underline g(x) , \]
where $\delta_x \in S(X)$ is the element determined by the promotion (in the LL sense)
of $x$ in the usual manner -- in this case this is simply the vector whose
$x$-th component is $1$ and the rest are $0.$ Therefore, following the standard
exponential discipline in models of LL, we can work with and recover the action
of the three aforementioned functions, provided that the inputs have norm at
most one. Of course, up to normalisation, this allows us to recover the action
of these maps on the entire domains.

\subsection{Higher-order Maps with Superposition}

The \emph{quantum switch} is a higher-order map that has attracted considerable
interest \cite{quantum-switch}. It admits a natural definition within
$\OS$ as a linear complete contraction
\begin{align}
  & \qsw \colon B(H) \ptimes B(H) \to B(\mathbb C^2 \htimes H) \label{eq:qsw-type} \\
  & \qsw(f \otimes g) \eqdef (\ket 0 \bra 0 \otimes (fg) ) + (\ket 1 \bra 1 \otimes (gf) ) \label{eq:qsw-def}
\end{align}
and it appears to make essential use of superposition (the ``$+$'' in its definition).
Definition \eqref{eq:qsw-def} of $\qsw$ makes implicit use of Proposition \ref{prop:universal-tensor}
which uniquely determines the map. To understand why $\qsw$ is a complete
contraction, we have to verify that assignment \eqref{eq:qsw-def} is
bilinear and that it satisfies the required norm condition. This is the
case (Appendix \ref{app:quantum-primitives}), so $\qsw$ is
a higher-order map in $\OS$.

Operator spaces have a well-developed multilinear decomposition
theory \cite[\S 9.4]{effros-ruan} which allows us to approach difficult questions such as
whether a function makes essential use of superposition or not.
This can be achieved through the use of a different tensor product, called the
\emph{Haagerup tensor product} \cite[9.2]{effros-ruan},
\cite[1.5.4]{blecher-merdy}, and we briefly sketch how this works, see Appendix
\ref{app:quantum-primitives} for more details.

If $X$ and $Y$ are operator spaces, the Haagerup tensor product
$X \hagtimes Y$ is again an operator space such that the
algebraic tensor product $X \otimes Y$ is dense in $X \hagtimes Y.$
By using standard results, it is straightforward to prove that
$\OS$ becomes a monoidal category when equipped with $(\cdot \hagtimes \cdot)$
as tensor product. If $v \in \mathbb M_n(X \otimes Y)$, then
\[ \norm{v}_h \leq \norm{v}_{\wedge} \qquad \cite[\textrm{Proposition }1.5.13]{blecher-merdy} \]
where $\norm{\cdot}_h$ is the associated norm on $M_n(X \hagtimes Y)$.
It follows that there is a complete contraction
\[
  \iota \colon X \ptimes Y \to X \hagtimes Y \qquad \cite[\textrm{Proposition 1.5.13}]{blecher-merdy}
\]
which coincides with the identity on elementary tensors and which determines the map.
This means that it is \emph{easier} to extend a linear map
$f \colon X \otimes Y \to Z$ to a complete contraction
$X \ptimes Y \to Z$ compared to a complete contraction
$X \hagtimes Y \to Z.$ 
Indeed, if we attempt to define the quantum switch with respect to the Haagerup tensor
\begin{align}
  & \qsw' \colon B(H) \hagtimes B(H) \to B(\mathbb C^2 \htimes H) \label{eq:qsw-type-hag} \\
  & \qsw'(f \otimes g) \eqdef (\ket 0 \bra 0 \otimes (fg) ) + (\ket 1 \bra 1 \otimes (gf) ) \label{eq:qsw-def-hag}
\end{align}
then we see that the map $\qsw'$ is \emph{not} completely contractive (see
Appendix \ref{app:quantum-primitives}) and therefore it is not a morphism
in $\OS.$ In other words, $\qsw$ does not factorise through $\iota.$
This is important, as far as multilinear decompositions are concerned, because
the Haagerup tensor enjoys the following property.

\begin{restatable}{proposition}{hagDecompose}
  \label{prop:hag-decompose}
  Given operator spaces $X_1$ and $X_2$ and Hilbert spaces $H$ and $L$,
  a linear map $\varphi \colon X_1 \hagtimes X_2 \to B(H, L)$ is
  completely bounded (completely contractive) iff there exists a Hilbert
  space $K$ and completely bounded (completely contractive) maps $\psi_1 \colon X_1 \to B(K, L)$ and $\psi_2 \colon X_2 \to B(H,K)$
  such that
  \[ \varphi(x_1 \otimes x_2) = \psi_1(x_1) \psi_2(x_2) \]
  and in this case one can choose $\psi_1$ and $\psi_2$ such that $\norm{\varphi}_{\mathrm{cb}} = \norm{\psi_1}_{\mathrm{cb}} \norm{\psi_2}_{\mathrm{cb}} . $
\end{restatable}
\begin{proof}
  The version of this proposition for completely bounded maps is proven in
  \cite[Theorem 9.4.3]{effros-ruan}. The version with the completely
  contractive maps follows as a simple corollary, see Appendix
  \ref{app:quantum-primitives}.
  \end{proof}

More generally, a complete contraction
$\varphi \colon X \ptimes Y \to B(H)$ admits a multilinear
decomposition like the one in Proposition \ref{prop:hag-decompose} iff
it factorises through $\iota$ as in the following diagram.
\cstikz{hag-factor.tikz}
The quantum switch $\qsw$ does not factorise as above and it does not
admit a suitable multilinear decomposition.
This observation gives further credence to the claim that
the superposition that is part of $\qsw$'s
definition appears to be essential.

Multilinear decompositions of maps, in the sense described here, are outside of
the usual remit of LL. Indeed, the Haagerup tensor cannot be
understood using LL alone. We discuss an extension of MLL
(with Mix) that would better match the Haagerup tensor in Section
\ref{sec:discuss}.

\section{Linear Logic and the Heisenberg-\Schrod{} Duality}
\label{sec:ll-hs}

The category $\OS$ is a model of ILL, but not of CLL, because the canonical
complete isometry $X \hookrightarrow X^{**}$ need not be an isomorphism. We now
explain how to build a model of CLL, based on $\OS$, whose duality is
compatible with the Heisenberg-\Schrod{} duality of quantum theory. A
distinguishing feature of our model is that the polarised multiplicative
connectives of LL correspond to two different tensor products that describe
spacewise system composition in the Heisenberg and \Schrod{} pictures,
respectively, and this works in infinite dimensions as well.

In two papers \cite{barr-autonomous-ll,barr-accessible-ll}, Barr showed how the
Chu construction \cite{Chu} can be used to build a model of CLL provided that
some strong conditions are satisfied. We show how this can be applied to
$\OS.$

\begin{definition}
  Let $\QQ$ be the category given by the following data:
  objects are triples $(X,Y,d)$ where $X$ and $Y$ are operator spaces and $d \colon X \ptimes Y \to \mathbb C$
  is a complete contraction; a morphism is a pair $(f, g) \colon (X_1, Y_1, d_1) \to (X_2, Y_2, d_2)$
  where $f \colon X_1 \to X_2$ and $g \colon Y_2 \to Y_1$ are complete contractions such that the following diagram
  \begin{equation}
    \label{eq:chu-def}
    \stikz{chu-def.tikz}
  \end{equation}
  commutes; composition
  and identities are defined as in $\OS.$ The category $\QQ$ is precisely the category $\mathbf{Chu}(\OS,\mathbb C)$,
  i.e. the Chu category of $\OS$ with dualising object given by $\mathbb C.$
\end{definition}

\begin{theorem}
  The category $\QQ$ is complete, cocomplete, $*$-autonomous, has a Lafont exponential
  and it is therefore a model of full linear logic.
\end{theorem}
\begin{proof}
  This follows immediately from \cite{barr-accessible-ll}, because $\OS$ is locally presentable, symmetric
  monoidal closed and the tensor unit $\mathbb C$ is an internal cogenerator.
  Note that the proof uses the Lafont exponential $! \colon \OS \to \OS$ from \eqref{eq:cocommutative-cofree}.
\end{proof}

We omit some of the technical details of how the relevant data is defined in
$\QQ$. Instead, we focus on aspects of $\QQ$ that have a clear
relevance to the Heisenberg-\Schrod{} duality.

The dual in $\QQ$, in the sense of $*$-autonomy, is defined as $(X,Y,d)^\perp \eqdef (Y, X, d \circ \sigma)$
on objects, where $\sigma \colon Y \ptimes X \xrightarrow{\cong} X \ptimes Y$ is the symmetry. On morphisms,
it is simply $(f,g)^\perp \eqdef (g,f).$ There is a canonical full and faithful embedding of $\OS$ into $\QQ$
which, up to natural isomorphism, and as a special case, allows us to pair up descriptions of quantum systems in the Heisenberg and \Schrod{} pictures
of quantum theory. More specifically, for a Hilbert space $H$, the bilinear form \eqref{eq:hs-type} is jointly complete contractive \cite[(7.1.12)]{effros-ruan}, hence it may be identified
with a complete contraction $\trace \colon T(H) \ptimes B(H) \to \mathbb C$. Then a pair
\begin{equation}
  \label{eq:hs-q}
  (f,g) \colon (T(H_1), B(H_1), \trace) \to (T(H_2),B(H_2),\trace)
\end{equation}
is a morphism in $\QQ$ iff $g = f^t$ in the sense of Proposition \ref{prop:hs-correspondence}.
To see this, observe that the universal property \eqref{eq:hs-dual-system} of the dual pairing that
determines the Heisenberg-\Schrod{} duality is equivalent to the commutativity of diagram \eqref{eq:chu-def}.
Therefore, in \eqref{eq:hs-q}, if $f$ is a CPTP map, then $g = f^t$, which is the corresponding NCPU map
in the Heisenberg picture and the duality $(-)^\perp$ in $\QQ$, in this case, can be equivalently understood
as applying the Heisenberg-\Schrod{} duality in the sense of
\eqref{eq:hs-duality}.

The arguments we presented in this
section, so far, apply equally well to the category $\Ban$.
However, multiple problems with $\Ban$ arise when we extend the scope of
discourse. First, transposition of matrices gives an isometric isomorphism
$(-)^T \colon M_n \cong M_n$ in $\Ban$, but this operation is \emph{not}
physically admissible when $n \geq 2$ and there is no clear way how to deal
with this problem using Banach spaces and the Chu construction. This map is not
completely contractive or completely positive when $n \geq 2$ which justifies
using methods from noncommutative geometry. Secondly, it is unclear how to
make the correspondence between the multiplicative conjunction/disjunction when using $\Ban$ (and
the Chu construction) and the appropriate way to compose systems in quantum theory.
However, when using $\OS$ and $\QQ$, this can be achieved.
Given Hilbert spaces $H_1$ and $H_2$,
we have
\begin{equation}
  \label{eq:ptimes-good}
  T(H_1) \ptimes T(H_2)    \cong     T(H_1 \htimes H_2)
\end{equation}
as operator spaces \cite[Proposition 7.2.1]{effros-ruan}. Quantum systems in the \Schrod{} picture
are indeed composed as in \eqref{eq:ptimes-good}. In the Heisenberg picture, they are composed
via the spatial tensor product of von Neumann algebras, and we explain how this arises in $\QQ.$
For objects of the form $(T(H), B(H), \trace)$, the monoidal structure of $\QQ$
allows us to recognise, in the sense of Polarised Linear Logic \cite{pll}, the
completely projective tensor $\ptimes$ as the multiplicative conjunction, and
the spatial tensor product $\stimes$ of von Neumann algebras\footnote{Viewed as dual operator spaces.} as the
multiplicative disjunction. To understand why that is the case, we recall that if $M$ and $N$
are von Neumann algebras, and if we write $M_*$ for the predual operator space of $M$ (guaranteed to exist),
then we have a completely isometric isomorphism
\[ (M \stimes N)_* \cong M_* \ptimes N_*  \qquad \text{\cite[Theorem 7.2.4]{effros-ruan}} \]
and by taking the dual of this isomorphism we see that
\begin{equation}
  \label{eq:spatial}
  M \stimes N \cong (M_* \ptimes N_*)^*
\end{equation}
which is a (well-known) way to describe the spatial tensor product, as far as
its structure as a dual operator space is concerned.

\begin{proposition}
  \label{prop:product-q}
  The monoidal product
  \[ (T(H_1), B(H_1), \trace) \otimes (T(H_2),B(H_2),\trace) \]
  in $\QQ$ is the object
  \[ (T(H_1) \ptimes T(H_2), B(H_1) \stimes B(H_2), \trace'), \]
  where $\trace'$ may be defined through the isomorphisms $T(H_1) \ptimes T(H_2) \cong T(H_1 \htimes H_2)$
  and $B(H_1) \stimes B(H_2) \cong B(H_1 \htimes H_2)$ and the canonical dual pairing on the latter spaces.
\end{proposition}
\begin{proof}
  Using \cite[Proposition 7.1]{barr-autonomous-ll} it follows immediately that
  \[ (X,X^*, \mathrm{ev}) \otimes (Y,Y^*, \mathrm{ev}) \cong (X \ptimes Y, (X \ptimes Y)^*, \mathrm{ev}) \]
  in $\QQ,$ where $\mathrm{ev}$ is given by evaluation.
  The proof follows using the isomorphisms $T(H)^* \cong B(H)$ and \eqref{eq:spatial}.
\end{proof}

The above proposition shows how to recover the spatial tensor product
(composition in the Heisenberg picture) from the completely projective one
(composition in the Schrödinger picture) via the Chu construction (coming from
the semantics of LL), possibly involving infinite-dimensional Hilbert spaces.
We think that this is an interesting result that lies in the intersection of
noncommutative geometry, von Neumann algebras, and semantics of LL. We hope
that this new insight can lead to a better understanding of the
Heisenberg-\Schrod{} duality from the point of view of LL and semantics.

Moving to the additives, the situation is completely analogous to the Banach
space case: the additive conjunction (disjunction) corresponds to the
$\ell^\infty$ $(\ell^1)$ direct sum, as already recognised by Girard
\cite{girard-cbs}, and the same can be said for $\QQ$, see also Proposition
\ref{prop:products} and Proposition \ref{prop:coproducts} which show that this
is how (co)products are constructed in $\OS$.

\begin{table}[t]
  \begin{center}
  \begin{tabular}{ |c|c|c| } 
    \hline
    \Schrod{} Picture & $\QQ$ & $\text{LL}_{+}$  \\
    \hline
    System description & $T(H_P)$ & $P$ \rule{0pt}{2ex} \\
    \hline
    & $T(H_P) \ptimes T(H_R) $ & \rule{0pt}{2.5ex} \\
    Quantum composition & $ \cong  $ & $P \otimes R$ \\
    & $ T(H_P \htimes H_R) $ &\\
   \hline
    Classical composition & $T(H_P) \loneplus T(H_R)$ & $P \oplus R$ \\ 
   \hline
  \end{tabular}
  \end{center}
  \caption{\Schrod{} picture and positive logical polarity.}
  \label{tab:schrod}
\end{table}

Let us explain how we can understand this situation using ideas from Polarised
Linear Logic. If we use $P,R$ to range over atomic formulas with positive
logical polarity, to which we associate Hilbert spaces $H_P,H_R$, the
operator spaces $T(H_P)$ and $T(H_R)$ contain the \emph{states} relevant to the \Schrod{} picture.
We use $N,M$ to range over atomic
formulas with negative logical polarity, to which we again associate Hilbert
spaces, and now the von Neumann algebras $B(H_N)$ and $B(H_M)$
contain the \emph{observables} relevant to the Heisenberg picture.
The Heisenberg-\Schrod{}
duality in physical contexts is often presented as the correspondence between
states and observables, hence the emphasis in the preceding sentences. Consider
Tables \ref{tab:schrod} and \ref{tab:heis}. By \emph{"Quantum composition"} we
mean the space-wise composition of quantum systems, where both quantum and
classical (i.e. non-quantum) interactions are possible between the two
subsystems. By \emph{"Classical composition"} we mean space-wise composition
where only classical interactions are possible between the two subsystems. In
order to understand intuitively why that is the case, observe that if we take
the idea of classical composition to the extreme, we have that
\[
  \ell^\infty(X) \cong \bigoplus_{X}^{\infty} \mathbb C
\]
is a \emph{commutative} von Neumann algebra, which is well-known to represent
classical information in the Heisenberg picture. Its predual $\ell^1(X)$
can be used to represent classical information in the \Schrod{} picture.

In fact, since von Neumann algebras are closed under $\ell^\infty$ direct sums,
and their preduals under $\ell^1$ direct sums, Tables \ref{tab:schrod} and
\ref{tab:heis} can be summarised (and generalised) by saying: von Neumann
algebras (viewed as dual operator spaces) correspond to formulas with negative
logical polarity and their preduals correspond to formulas with positive
logical polarity. This further reinforces the view that we suggested: 
\begin{align}
  \text{\Schrod{} picture }  &\Rightarrow \text{ positive logical polarity} \label{eq:schrod} \\
  \text{Heisenberg picture } &\Rightarrow \text{ negative logical polarity} \label{eq:heis}
\end{align}

\begin{table}[t]
  \begin{center}
  \begin{tabular}{ |c|c|c| }
    \hline
    Heisenberg Picture & $\QQ$ & $\text{LL}_{-}$ \\
    \hline
    System description & $B(H_N)$ & $N$ \rule{0pt}{2ex} \\
    \hline
    & $B(H_N) \stimes B(H_M) $ & \rule{0pt}{2.5ex} \\
    Quantum composition & $\cong$ & $N \parr M$ \\
    & $ B(H_N \htimes H_M) $ & \\
    \hline
    Classical composition & $B(H_N) \linftyplus B(H_M)$ & $N \aconj M$ \\
   \hline
  \end{tabular}
  \end{center}
  \caption{Heisenberg picture and negative logical polarity.}
  \label{tab:heis}
\end{table}

\section{Discussion and Related Work}
\label{sec:discuss}

We gave strong evidence in support of the view that the Heisenberg-\Schrod{}
duality can be studied via methods from linear logic. However, much work
remains to be done. Tables \ref{tab:schrod} and \ref{tab:heis} show that we get
the duality in a reasonable sense on the object level, but \textbf{CPTP} (or \textbf{NCPU}) is
not a \emph{full} subcategory of $\OS$ or $\QQ.$ However, we do have a faithful
embedding and the duality in $\QQ$ does indeed coincide with
\eqref{eq:hs-duality} for those morphisms that are CPTP/NCPU.
Another class of morphisms that are of interest are the completely positive
trace-non-increasing (CPTNI) maps $T(H_1) \to T(H_2)$ and the normal completely
positive subunital (NCPSU) maps $B(H_2) \to B(H_1)$. These classes of maps are
also important, especially if we want to model general recursion, and the
Heisenberg-\Schrod{} duality \eqref{eq:hs-duality} can be extended to provide a
bijective correspondence between them, so we get another isomorphism of
categories $\CPTNI \cong \NCPSU^{\mathrm{op}}.$ Just like $\CPTP$ ($\NCPU$), we
have a faithful, but not full, embedding of $\CPTNI$ ($\NCPSU$) into $\OS$ and
$\QQ.$ This means that there are undesirable maps in the models from the point
of view of mixed state quantum computation.

It is important to restrict the relevant homsets in the models to CPTP/NCPU (or
CPTNI/NCPSU if recursion is desired) and we leave that for future work. There
are established techniques in the literature that can potentially allow us to
restrict to the desired maps, e.g. methods based on gluing and orthogonality
\cite{hyland-schalk} show how to restrict the homsets of models of ILL and CLL
while still retaining the ILL and CLL structure. This is indeed promising and
these ideas have already been used to similar effect in a quantum setting
(without exponentials) in \cite{aleks-sander,quantum-caus-bv} where the authors
restrict a model based on CP maps in such a way that the restricted category
has a fully faithful embedding of $\CPTP.$ Note that, if we would add
trace-preservation (TP), as in \cite{aleks-sander,quantum-caus-bv}, on the
relevant homsets of our models, then Proposition \ref{prop:cc-cp-tp} implies
that the maps would be CPTP. Restricting to CPTNI/NCPSU (for recursion) is also
future work, where we hope to use similar ideas. What is more challenging is to
restrict the models in such a way that the noncommutative structure of operator
spaces is preserved to a sufficient degree that allows us to talk about the
Haagerup tensor, but it is not clear what is the best way to do so yet.

The isomorphism $T(H_1) \ptimes T(H_2) \cong T(H_1 \htimes
H_2)$ suggests that the completely projective tensor is well-behaved in the
\Schrod{} picture. Such an isomorphism is not known to hold for the projective
tensor of Banach spaces (to the best of our knowledge). A category defined by
Selinger \cite{selinger-towards}, based on ``normed cones'', was shown to be
$*$-autonomous, but its monoidal product is unsatisfactory for quantum
computation. Intuitively, the reason $\ptimes$ behaves better is because it
takes into account the \emph{noncommutative} structure and so it
behaves more appropriately with respect to the amplifications of operators.

The noncommutative structure of $\OS$ is essential for the description of the
Haagerup tensor and its interesting properties (\secref{sec:pure-mixed}).
This tensor is \emph{not} symmetric and for \emph{finite-dimensional} operator
spaces we have $X^* \hagtimes Y^* \cong (X \hagtimes Y)^*$ \cite[pp.
169]{effros-ruan}, thus mimicking closely the non-symmetry and self-duality of
the \emph{seq} connective from BV-logic \cite{bv-logic}. In fact, Thea Li
proved \cite{thea-internship} that the subcategory of \emph{finite-dimensional}
operator spaces is a BV-category \cite{bv-category} with respect to the
Haagerup tensor, and a model of MALL, but unsurprisingly, it appears there is
no way to model exponentials there. BV-logic does not admit a sequent calculus
presentation \cite{bv-logic-sequent} and there is no proof in the literature
that a BV-category is a \emph{sound} model of BV-logic, but some of the
important structure is indeed captured. For (infinite-dimensional) operator
spaces in $\OS$, we have a complete isometry $X^* \hagtimes Y^* \hookrightarrow
(X \hagtimes Y)^*$ \cite[pp. 168]{effros-ruan} and it would be interesting to
investigate how other tensors related to the Haagerup one (e.g. extended
Haagerup tensor) \cite{hopf-operator} would behave in $\QQ$ and any potential
connections to BV-logic. Connections between BV-categories and ``causality''
have already been studied in a quantum context
\cite{aleks-sander,quantum-caus-bv,caus3} with respect to completely-positive
maps. Note that $\OS$ is \emph{not} a BV-category, but the Haagerup tensor can
be used to reason about multilinear decompositions of complete contractions.

The category $\QQ$ clearly has a resemblance to the Coherent Banach
Spaces of Girard \cite{girard-cbs}. One of the major differences is that we work in a
noncommutative (or quantum) setting which complicates the mathematical
development. It would be interesting to say more about this relationship,
especially with respect to the $\Min$ quantisation functor and Corollary
\ref{cor:ban-os}. Unfortunately, we lack a convenient description of the Lafont exponential
based on constructions from operator space theory. Such a
description would bring us closer to being able to add an extra row to Tables
\ref{tab:schrod} and \ref{tab:heis} with the missing exponentials and being
able to add the inverse implications in \eqref{eq:schrod} and \eqref{eq:heis},
hopefully relating this to existing constructions in the operator space literature.

We focused on models, but designing appropriate logics/type systems is also
relevant. Full abstraction results for quantum lambda calculi are already known
and include models based on game semantics
\cite{quantum-game,qlc-full-abstraction,marc-phd}, enriched presheaves
\cite{qfpc}, and quantitative semantics of linear logic \cite{qlc-quantitative}
(full abstraction proven later in \cite{qlc-full-abstraction,marc-phd}). We
hope that our work can also contribute towards such results, and also language
design, after restricting to CPTNI/NCPSU maps. Note that the model in
\cite{qfpc} also starts with a locally presentable model of ILL out of which
the authors carve out a model of CLL.

Note that if one succeeds in building a categorical model of CLL whose duality coincides with
the Heisenberg-\Schrod{} duality, together with a fully faithful strong symmetric monoidal embedding
of $\CPTP/\NCPU$ (or $\CPTNI/\NCPSU$), this would be very natural from a quantum perspective, because many of the (polarised) formulas admit natural descriptions in quantum theory,
as we already showed in Tables \ref{tab:schrod} and \ref{tab:heis}. Furthermore, since the
linear implication in CLL is decomposed as $A \multimap B \equiv A^\perp \parr B$, this suggests
that the higher-order structure would also be induced in a natural way.
More generally, something that distinguishes our approach to semantics from the cited
related works, is that we use mathematics that is more traditional
to quantum theory in order to describe the higher-order structure of the
models.

To conclude, our paper highlights interesting connections between linear logic
and quantum theory by using results from noncommutative geometry
and operator spaces. We hope that our work can serve as a foundation and
building block for further semantic developments related to quantum theory.

\vspace{1mm}
\noindent\textbf{Acknowledgements.} We thank the anonymous reviewers for their
feedback which led to multiple improvements of the paper. We also thank James
Hefford, Timothée Hoffreumon, Anna Jenčová, Thea Li, Jean-Simon Pacaud Lemay,
Jiří Rosický, and Matthew Wilson for discussions and/or useful feedback. This
work has been partially funded by the Austrian Science Fund (FWF) under the
research project PAT6443523 (\href{https://doi.org/10.55776/PAT6443523}{DOI: 10.55776/PAT6443523}) and by the  French
National Research Agency (ANR) within the framework of ``Plan France 2030'',
under the research projects EPIQ ANR-22-PETQ-0007, HQI-Acquisition
ANR-22-PNCQ-0001, HQI-R\&D ANR-22-PNCQ-0002, and also by CIFRE 2022/0081.

\newpage

\bibliographystyle{IEEEtran}
\bibliography{refs}

\begin{thebibliography}{10}
\providecommand{\url}[1]{#1}
\csname url@samestyle\endcsname
\providecommand{\newblock}{\relax}
\providecommand{\bibinfo}[2]{#2}
\providecommand{\BIBentrySTDinterwordspacing}{\spaceskip=0pt\relax}
\providecommand{\BIBentryALTinterwordstretchfactor}{4}
\providecommand{\BIBentryALTinterwordspacing}{\spaceskip=\fontdimen2\font plus
\BIBentryALTinterwordstretchfactor\fontdimen3\font minus
  \fontdimen4\font\relax}
\providecommand{\BIBforeignlanguage}[2]{{%
\expandafter\ifx\csname l@#1\endcsname\relax
\typeout{** WARNING: IEEEtran.bst: No hyphenation pattern has been}%
\typeout{** loaded for the language `#1'. Using the pattern for}%
\typeout{** the default language instead.}%
\else
\language=\csname l@#1\endcsname
\fi
#2}}
\providecommand{\BIBdecl}{\relax}
\BIBdecl

\bibitem{pedersen:analysisnow}
G.~Pedersen, \emph{{Analysis Now}}.\hskip 1em plus 0.5em minus 0.4em\relax
  Springer, 1989.

\bibitem{connes:ncg}
A.~Connes, \emph{Noncommutative Geometry}.\hskip 1em plus 0.5em minus
  0.4em\relax Academic Press, 1994.

\bibitem{Blackadar17}
B.~Blackadar, ``Operator algebras: Theory of {C}*-algebras and von neumann
  algebras,'' 2006.

\bibitem{kadisonringrose:oa1}
R.~Kadison and J.~Ringrose, \emph{Fundamentals of the Theory of Operator
  Algebra, Volume {I}: Elementary Theory}.\hskip 1em plus 0.5em minus
  0.4em\relax American Mathematical Society, 1997.

\bibitem{kadisonringrose:oa2}
------, \emph{Fundamentals of the Theory of Operator Algebra, Volume {II}:
  Advanced Theory}.\hskip 1em plus 0.5em minus 0.4em\relax American
  Mathematical Society, 1997.

\bibitem{takesaki}
M.~Takesaki, \emph{{Theory of Operator Algebras. Vol. I, II and III}}.\hskip
  1em plus 0.5em minus 0.4em\relax Springer-Verlag, Berlin, 2002.

\bibitem{linear-logic}
J.-Y. Girard, ``Linear logic,'' \emph{Theoretical computer science}, vol.~50,
  no.~1, pp. 1--101, 1987.

\bibitem{pll}
\BIBentryALTinterwordspacing
O.~Laurent, ``{{\'E}tude de la polarisation en logique},'' Theses,
  {Universit{\'e} de la M{\'e}diterran{\'e}e - Aix-Marseille II}, Mar. 2002.
  [Online]. Available: \url{https://theses.hal.science/tel-00007884}
\BIBentrySTDinterwordspacing

\bibitem{pll-fixpoints}
\BIBentryALTinterwordspacing
T.~Ehrhard, F.~Jafarrahmani, and A.~Saurin, ``{Polarized Linear Logic with
  Fixpoints},'' {IRIF (UMR\_8243) - Institut de Recherche en Informatique
  Fondamentale}, Technical Report, Apr. 2022. [Online]. Available:
  \url{https://hal.science/hal-03655737}
\BIBentrySTDinterwordspacing

\bibitem{mellies-linear-logic}
P.-A. Mellies, ``Categorical semantics of linear logic,'' \emph{Panoramas et
  synth{\`e}ses-Soci{\'e}t{\'e} math{\'e}matique de France}, no.~27, 2009.

\bibitem{seely-linear}
R.~Seely, ``Linear logic,*-autonomous categories and cofree coalgebras,''
  \emph{Categories in Computer Science and Logic}, vol.~92, pp. 371--382, 1989.

\bibitem{barr-autonomous-book}
M.~Barr, \emph{*-Autonomous categories}.\hskip 1em plus 0.5em minus 0.4em\relax
  Springer, 2006, vol. 752.

\bibitem{effros-ruan}
E.~G. Effros and Z.-J. Ruan, \emph{Theory of Operator Spaces}.\hskip 1em plus
  0.5em minus 0.4em\relax American Mathematical Society, 2022, vol. 386.

\bibitem{blecher-merdy}
\BIBentryALTinterwordspacing
D.~P. Blecher and C.~Le~Merdy, \emph{{Operator Algebras and Their Modules: An
  operator space approach}}.\hskip 1em plus 0.5em minus 0.4em\relax Oxford
  University Press, 10 2004. [Online]. Available:
  \url{https://doi.org/10.1093/acprof:oso/9780198526599.001.0001}
\BIBentrySTDinterwordspacing

\bibitem{pisier}
G.~Pisier, \emph{Introduction to operator space theory}.\hskip 1em plus 0.5em
  minus 0.4em\relax Cambridge University Press, 2003.

\bibitem{lp-categories-book}
J.~Adamek and J.~Rosick{\`y}, \emph{Locally presentable and accessible
  categories}.\hskip 1em plus 0.5em minus 0.4em\relax Cambridge University
  Press, 1994, vol. 189.

\bibitem{quantum-switch}
G.~Chiribella, G.~M. D’Ariano, P.~Perinotti, and B.~Valiron, ``Quantum
  computations without definite causal structure,'' \emph{Physical Review A},
  vol.~88, no.~2, p. 022318, 2013.

\bibitem{bv-logic}
\BIBentryALTinterwordspacing
A.~Guglielmi, ``A system of interaction and structure,'' \emph{ACM Trans.
  Comput. Logic}, vol.~8, no.~1, p. 1–es, jan 2007. [Online]. Available:
  \url{https://doi.org/10.1145/1182613.1182614}
\BIBentrySTDinterwordspacing

\bibitem{Chu}
P.-H. Chu, ``Constructing *-autonomous categories,'' in \emph{*-Autonomous
  Categories}, ser. Lecture Notes in Mathematics, M.~Barr, Ed.\hskip 1em plus
  0.5em minus 0.4em\relax Springer, 1979, vol. 752, pp. 103--138, appendix to:
  Michael Barr, *-Autonomous Categories.

\bibitem{gabriel-ulmer}
\BIBentryALTinterwordspacing
P.~Gabriel and F.~Ulmer, \emph{Lokal pr{\"a}sentierbare kategorien}, ser.
  Lecture Notes in Mathematics.\hskip 1em plus 0.5em minus 0.4em\relax
  Springer-Verlag, 2006, vol. 221. [Online]. Available:
  \url{https://doi.org/10.1007/BFb0059396}
\BIBentrySTDinterwordspacing

\bibitem{borceux2}
F.~Borceux, \emph{Handbook of Categorical Algebra: Volume 2, Categories and
  Structures}.\hskip 1em plus 0.5em minus 0.4em\relax Cambridge University
  Press, 1994, vol.~2.

\bibitem{finitary-functors}
\BIBentryALTinterwordspacing
J.~Ad{\'{a}}mek, S.~Milius, L.~S. Moss, and H.~Urbat, ``On finitary functors
  and their presentations,'' \emph{J. Comput. Syst. Sci.}, vol.~81, no.~5, pp.
  813--833, 2015. [Online]. Available:
  \url{https://doi.org/10.1016/j.jcss.2014.12.002}
\BIBentrySTDinterwordspacing

\bibitem{ban-separable}
\BIBentryALTinterwordspacing
J.~Adámek and J.~Rosický, ``Approximate injectivity and smallness in
  metric-enriched categories,'' \emph{Journal of Pure and Applied Algebra},
  vol. 226, no.~6, p. 106974, 2022. [Online]. Available:
  \url{https://www.sciencedirect.com/science/article/pii/S0022404921003157}
\BIBentrySTDinterwordspacing

\bibitem{OS-cat}
\BIBentryALTinterwordspacing
B.~Lindenhovius and V.~Zamdzhiev, ``The category of operator spaces and
  complete contractions,'' \emph{CoRR}, vol. abs/2412.20999, 2024. [Online].
  Available: \url{https://doi.org/10.48550/arXiv.2412.20999}
\BIBentrySTDinterwordspacing

\bibitem{borceux1}
F.~Borceux, \emph{Handbook of Categorical Algebra: Volume 1, Basic Category
  Theory}.\hskip 1em plus 0.5em minus 0.4em\relax Cambridge University Press,
  1994, vol.~1.

\bibitem{choi-adjunctions}
\BIBentryALTinterwordspacing
Y.~Choi, ``Natural examples of sequences of adjoint functors,'' MathOverflow,
  uRL:https://mathoverflow.net/q/89817 (version: 2012-02-29). [Online].
  Available: \url{https://mathoverflow.net/q/89817}
\BIBentrySTDinterwordspacing

\bibitem{pestov-adjunction}
\BIBentryALTinterwordspacing
V.~Pestov, ``Operator spaces and residually finite-dimensional c*-algebras,''
  \emph{Journal of Functional Analysis}, vol. 123, no.~2, pp. 308--317, 1994.
  [Online]. Available:
  \url{https://www.sciencedirect.com/science/article/pii/S0022123684710901}
\BIBentrySTDinterwordspacing

\bibitem{henry-lambda-adjunction}
\BIBentryALTinterwordspacing
S.~Henry, ``Combinatorial and accessible weak model categories,'' 2020.
  [Online]. Available: \url{https://arxiv.org/abs/2005.02360}
\BIBentrySTDinterwordspacing

\bibitem{choi-smcc}
\BIBentryALTinterwordspacing
Y.~Choi, ``Do completely bounded maps on an operator space have a completely
  contractive banach algebra structure?'' MathOverflow,
  uRL:https://mathoverflow.net/q/455809 (version: 2023-10-03). [Online].
  Available: \url{https://mathoverflow.net/q/455809}
\BIBentrySTDinterwordspacing

\bibitem{lafont-thesis}
Y.~Lafont, ``Logiques, cat{\'e}gories et machines,'' Ph.D. dissertation,
  Universit{\'e} Paris 7, 1988.

\bibitem{porst-comonoids}
H.-E. Porst, ``On categories of monoids, comonoids, and bimonoids,''
  \emph{Quaestiones Mathematicae}, vol.~31, no.~2, pp. 127--139, 2008.

\bibitem{benton-lnl}
\BIBentryALTinterwordspacing
P.~N. Benton, ``A mixed linear and non-linear logic: Proofs, terms and models
  (extended abstract),'' in \emph{Computer Science Logic, 8th International
  Workshop, {CSL} '94, Kazimierz, Poland, September 25-30, 1994, Selected
  Papers}, ser. Lecture Notes in Computer Science, L.~Pacholski and J.~Tiuryn,
  Eds., vol. 933.\hskip 1em plus 0.5em minus 0.4em\relax Springer, 1994, pp.
  121--135. [Online]. Available: \url{https://doi.org/10.1007/BFb0022251}
\BIBentrySTDinterwordspacing

\bibitem{hilb-contractions}
C.~Heunen, A.~Kornell, and N.~van~der Schaaf, ``Axioms for the category of
  hilbert spaces and linear contractions,'' \emph{Bull. London Math. Soc},
  vol.~56, pp. 1532--1549, 2024.

\bibitem{barr-autonomous-ll}
M.~Barr, ``*-autonomous categories and linear logic,'' \emph{Mathematical
  Structures in Computer Science}, vol.~1, no.~2, p. 159–178, 1991.

\bibitem{barr-accessible-ll}
\BIBentryALTinterwordspacing
------, ``Accessible categories and models of linear logic,'' \emph{Journal of
  Pure and Applied Algebra}, vol.~69, no.~3, pp. 219--232, 1991. [Online].
  Available:
  \url{https://www.sciencedirect.com/science/article/pii/0022404991900203}
\BIBentrySTDinterwordspacing

\bibitem{girard-cbs}
\BIBentryALTinterwordspacing
J.~Girard, ``Coherent banach spaces: a continuous denotational semantics,'' in
  \emph{Linear Logic Tokyo Meeting 1996, Keio University, Mita Campus, Tokyo,
  Japan, March 29 - April 2, 1996}, ser. Electronic Notes in Theoretical
  Computer Science, J.~Girard, M.~Okada, and A.~Scedrov, Eds., vol.~3.\hskip
  1em plus 0.5em minus 0.4em\relax Elsevier, 1996, pp. 81--87. [Online].
  Available: \url{https://doi.org/10.1016/S1571-0661(05)80405-0}
\BIBentrySTDinterwordspacing

\bibitem{hyland-schalk}
\BIBentryALTinterwordspacing
M.~Hyland and A.~Schalk, ``Glueing and orthogonality for models of linear
  logic,'' \emph{Theor. Comput. Sci.}, vol. 294, no. 1/2, pp. 183--231, 2003.
  [Online]. Available: \url{https://doi.org/10.1016/S0304-3975(01)00241-9}
\BIBentrySTDinterwordspacing

\bibitem{aleks-sander}
A.~Kissinger and S.~Uijlen, ``{A categorical semantics for causal structure},''
  in \emph{32nd Annual {ACM/IEEE} Symposium on Logic in Computer Science,
  {LICS} 2017, Reykjavik, Iceland, June 20-23, 2017}.\hskip 1em plus 0.5em
  minus 0.4em\relax {IEEE} Computer Society, 2017, pp. 1--12.

\bibitem{quantum-caus-bv}
\BIBentryALTinterwordspacing
W.~Simmons and A.~Kissinger, ``Higher-order causal theories are models of
  bv-logic,'' in \emph{47th International Symposium on Mathematical Foundations
  of Computer Science, {MFCS} 2022, August 22-26, 2022, Vienna, Austria}, ser.
  LIPIcs, S.~Szeider, R.~Ganian, and A.~Silva, Eds., vol. 241.\hskip 1em plus
  0.5em minus 0.4em\relax Schloss Dagstuhl - Leibniz-Zentrum f{\"{u}}r
  Informatik, 2022, pp. 80:1--80:14. [Online]. Available:
  \url{https://doi.org/10.4230/LIPIcs.MFCS.2022.80}
\BIBentrySTDinterwordspacing

\bibitem{selinger-towards}
P.~Selinger, ``Towards a semantics for higher-order quantum computation,'' in
  \emph{Proceedings of the 2nd International Workshop on Quantum Programming
  Languages, TUCS General Publication}, vol.~33, 2004, pp. 127--143.

\bibitem{thea-internship}
T.~Li, ``The category of finite dimensional operator spaces,'' 2024, internship
  report (unpublished), defended in August 2024.

\bibitem{bv-category}
\BIBentryALTinterwordspacing
R.~Blute, P.~Panangaden, and S.~Slavnov, ``Deep inference and probabilistic
  coherence spaces,'' \emph{Appl. Categorical Struct.}, vol.~20, no.~3, pp.
  209--228, 2012. [Online]. Available:
  \url{https://doi.org/10.1007/s10485-010-9241-0}
\BIBentrySTDinterwordspacing

\bibitem{bv-logic-sequent}
\BIBentryALTinterwordspacing
L.~T.~D. Nguy{\^{e}}n and L.~Stra{\ss}burger, ``A system of interaction and
  structure {III:} the complexity of {BV} and pomset logic,'' \emph{Log.
  Methods Comput. Sci.}, vol.~19, no.~4, 2023. [Online]. Available:
  \url{https://doi.org/10.46298/lmcs-19(4:25)2023}
\BIBentrySTDinterwordspacing

\bibitem{hopf-operator}
E.~G. Effros and Z.-J. Ruan, ``Operator space tensor products and hopf
  convolution algebras,'' \emph{Journal of Operator Theory}, pp. 131--156,
  2003.

\bibitem{caus3}
\BIBentryALTinterwordspacing
W.~Simmons and A.~Kissinger, ``A complete logic for causal consistency,''
  \emph{CoRR}, vol. abs/2403.09297, 2024. [Online]. Available:
  \url{https://doi.org/10.48550/arXiv.2403.09297}
\BIBentrySTDinterwordspacing

\bibitem{quantum-game}
P.~Clairambault, M.~de~Visme, and G.~Winskel, ``{Game semantics for quantum
  programming},'' \emph{{PACMPL}}, vol.~3, no. {POPL}, pp. 32:1--32:29, 2019.

\bibitem{qlc-full-abstraction}
\BIBentryALTinterwordspacing
P.~Clairambault and M.~de~Visme, ``Full abstraction for the quantum
  lambda-calculus,'' \emph{Proc. {ACM} Program. Lang.}, vol.~4, no. {POPL}, pp.
  63:1--63:28, 2020. [Online]. Available: \url{https://doi.org/10.1145/3371131}
\BIBentrySTDinterwordspacing

\bibitem{marc-phd}
\BIBentryALTinterwordspacing
M.~de~Visme, ``Quantum game semantics. (s{\'{e}}mantique des jeux quantique),''
  Ph.D. dissertation, University of Lyon, France, 2020. [Online]. Available:
  \url{https://tel.archives-ouvertes.fr/tel-03045844}
\BIBentrySTDinterwordspacing

\bibitem{qfpc}
\BIBentryALTinterwordspacing
T.~Tsukada and K.~Asada, ``Enriched presheaf model of quantum {FPC},''
  \emph{Proc. {ACM} Program. Lang.}, vol.~8, no. {POPL}, pp. 362--392, 2024.
  [Online]. Available: \url{https://doi.org/10.1145/3632855}
\BIBentrySTDinterwordspacing

\bibitem{qlc-quantitative}
\BIBentryALTinterwordspacing
M.~Pagani, P.~Selinger, and B.~Valiron, ``Applying quantitative semantics to
  higher-order quantum computing,'' in \emph{The 41st Annual {ACM}
  {SIGPLAN-SIGACT} Symposium on Principles of Programming Languages, {POPL}
  '14, San Diego, CA, USA, January 20-21, 2014}, S.~Jagannathan and P.~Sewell,
  Eds.\hskip 1em plus 0.5em minus 0.4em\relax {ACM}, 2014, pp. 647--658.
  [Online]. Available: \url{https://doi.org/10.1145/2535838.2535879}
\BIBentrySTDinterwordspacing

\bibitem{stormer}
E.~St{\o}rmer, \emph{Positive Linear Maps of Operator Algebras}.\hskip 1em plus
  0.5em minus 0.4em\relax Springer, 2013.

\bibitem{HeinosaariZiman}
\BIBentryALTinterwordspacing
T.~Heinosaari and M.~Ziman, \emph{The Mathematical Language of Quantum Theory:
  From Uncertainty to Entanglement}.\hskip 1em plus 0.5em minus 0.4em\relax
  Cambridge: Cambridge University Press, 2011, online publication date: January
  2012. [Online]. Available: \url{https://doi.org/10.1017/CBO9780511976515}
\BIBentrySTDinterwordspacing

\bibitem{Weidmann}
J.~Weidmann, \emph{{Linear Operators in Hilbert Spaces}}.\hskip 1em plus 0.5em
  minus 0.4em\relax Springer-Verlag New York Inc., 1980.

\end{thebibliography}

\newpage
\onecolumn
\appendices

\newpage
\section{Omitted Proofs from Section \ref{sec:os}}
\label{app:os}

In this section we provide more detailed proofs of some propositions from \secref{sec:os}.

\subsection{The Heisenberg-\Schrod{} Duality}
\label{app:hs}

\hsCorrespond*
\begin{proof}
For any Banach space $X$, there is an isometry $X\to X^{**}$ \cite[2.3.7]{pedersen:analysisnow}, whence the assignment $\psi\mapsto\psi^t$ is injective. The bijective correspondence between bounded linear maps $\psi:T(H_1)\to T(H_2)$ and normal linear maps now follows from \cite[Proposition 2.4.12]{pedersen:analysisnow}. 
For (1), positivity of $\psi^t$ implies positivity of $\psi$ by \cite[Proposition 1.4.2.(i)]{stormer}, whose proof is based on the observation for any Hilbert space $K$ that $b\in B(K)$ is positive if and only if $\mathrm{tr}(xb)\geq 0$ for each $x\in T(K)$. Now, assume that $\psi$ is positive. Let $b\in B(H_2)$ be positive. Then for each positive $x\in T(H_1)$, we have that $\psi(x)\in T(H_2)$ is positive, hence $\mathrm{tr}(\psi(x)b)$ is positive. So $\mathrm{tr}(x\psi^t(b))$ is positive. Since $x$ is arbitrary, it follows from the observation that $\psi^t(b)$ is positive. The statement for complete positivity follows from \cite[Section 4.1.2]{HeinosaariZiman}.
The proof of (2) is straightforward, and is exactly as in \cite[Proposition 1.4.2]{stormer}, even though the setting there is slightly different. 
It is shown in \cite[Section 2.3]{pisier} that $\psi$ is completely bounded if and only if $\psi^*$ (and equivalently, $\psi^t)$ is completely bounded, in which case their cb norms coincide. Hence, (3) follows.  
\end{proof}

\subsection{Proofs related to pure and mixed quantum primitives in $\OS$}
\label{app:quantum-primitives}

We start this section with some conventions. Given a bounded operator $f:H\to K$ between Hilbert spaces, we denote its Hilbert space adjoint by $f^\dag:K\to H$. Note that $\|f^\dag\|=\|f\|$ \cite[Theorem 4.14]{Weidmann}. We denote the tensor product of two Hilbert spaces $H$ and $K$ by  $H\htimes K$, whose inner product is determined by its action on elementary tensors in the following way:
\begin{equation}\label{eq:inner product tensor product}
    \langle h_1\otimes k_1,h_2\otimes k_2\rangle = \langle h_1,h_2\rangle\langle k_1,k_2\rangle.
\end{equation} 
As a consequence, we have $\|h\otimes k\|=\|h\|\|k\|$ for each $h\otimes k\in H\otimes K$, so the norm on $H\htimes K$ is a cross norm.

Furthermore, given Hilbert spaces $H_1, H_2,\ldots, H_n$, the inner product on $\bigoplus_{i=1}^nH_i$ is defined by 
\[\langle (h_1,\ldots,h_n),(h_1',\ldots, h_n')\rangle = \sum_{i=1}^n\langle h_i,h_i'\rangle,\]
whence 
\begin{equation}\label{eq:norm inner product Hilbert sum}
\|(h_1,\ldots,h_n)\|^2=\sum_{i=1}^n\|h_i\|^2    
\end{equation}
 for each $(h_1,\ldots,h_n)\in\bigoplus_{i=1}^nH_i$. We write $H^{\oplus n}\eqdef \bigoplus_{i=1}^nH$ for each Hilbert space $H$.

We denote the standard orthonormal basis of $\mathbb C^2$ by $\{\ket 0,\ket 1\}$. The orthogonal projections on $\ket 0$ and $\ket 1$ are denoted by $p_0$ and $p_1$, respectively. In the notation of physics, we have $p_0=\ket 0\bra 0$ and $p_1=\ket 1\bra 1$. Note that $p_i\ket i=\ket i$. The inner product of $\ket i$ and $\ket j$ is denoted by $\langle i|j\rangle$. Given a vector $h$ in another Hilbert space $H$, we write sometimes $\ket{i,h}$ instead of $\ket i\otimes h\in \mathbb C^2\htimes H$.

\begin{lemma}\label{lem:vectors in C2otimesH}
    Let $H$ be a Hilbert space. Then each $k\in\mathbb C^2\htimes H$ can be written as $k=\ket 0\otimes h_0+\ket 1\otimes h_1$ for some $h_0,h_1\in H$. Moreover, we have
\[\|k\|^2 = \|h_0\|^2+\|h_1\|^2.\]
\end{lemma}
\begin{proof}
Let $(f_\alpha)_{\alpha\in A}$ be an orthonormal basis for $H$. Then elementary tensors of the form $\ket i\otimes f_\alpha$ form an orthonormal basis for $\mathbb C^2\htimes H$. Hence, we can write $k=\sum_{\alpha\in A}\lambda_{0,\alpha}\ket 0\otimes f_\alpha+\lambda_{1,\alpha}\ket 1\otimes f_\alpha$ for some $\lambda_{i,\alpha}\in\mathbb C$. Then indeed $k=\ket 0\otimes h_0+\ket 1\otimes h_1$ with $h_i\eqdef\sum_{\alpha\in A}\lambda_{i,\alpha}f_\alpha$.
Since $\ket 0$ and $\ket 1$ form an orthonormal basis for $\mathbb C^2$, we have $\langle  0| 1\rangle=0$. 
It now follows from equation (\ref{eq:inner product tensor product}) that $\langle   {0,h_0}|1, h_1\rangle=0$, where $\ket{i,h_i}\eqdef\ket i\otimes h_i$. Using Parseval's identity and the fact that the norm on $\mathbb C^2\htimes H$ is a cross norm yields
\[ \|k\|^2 =\|\ket 0\otimes h_0\|^2+\|\ket 1\otimes h_1\|^2=\|\ket 0\|^2\|h_0\|^2+\|\ket 1\|^2\|h_1\|^2=\|h_0\|^2+\|h_1\|^2. \qedhere\]
\end{proof}

\begin{proposition}
    Let $H$ and $K$ be Hilbert spaces. Then we can restrict and corestrict the following functions to the unit balls of their domains and codomains:
\begin{align*}
  & \Uadjoint \colon B(H,K) \to B(K,H) \\
  & \Uadjoint(f) \eqdef f^\dagger \\
  & \Uctrl \colon B(H) \to B(\mathbb C^2 \htimes H) \\
  & \Uctrl( f ) \eqdef (\ket 0 \bra 0 \otimes \id) + (\ket 1 \bra 1 \otimes f) \\
  & \Uapply \colon B(H,K) \to \CB( T(H), T(K)) \\
  & \Uapply(f) \eqdef f(\cdot)f^\dagger.
\end{align*}
\end{proposition}
\begin{proof}
Let $f\in \Ball(B(H,K))$, so $\|f\|\leq 1$. By \cite[Theorem 4.14]{Weidmann}, we have $\|f^\dagger\|=\|f\|\leq 1$, so $f^\dag\in\Ball(B(K,H))$. This shows that $\Uadjoint$ restricts and corestricts to a map $\Ball(B(H,K))\to\Ball(B(K,H))$.

Let $f\in \Ball(B(H))$, so $\|f\|\leq 1$. Let $k\in\mathbb C^2$. Using Lemma \ref{lem:vectors in C2otimesH}, we can write $k=\ket 0\otimes h_0+\ket 1\otimes h_1$ for some $h_0,h_1\in H$ with $\|k\|^2=\|h_0\|^2+\|h_1\|^2$. Then
\begin{align*}
    \|\Uctrl(f)k\|^2 & = \|(p_0\otimes \id_H +p_1\otimes f)(\ket 0\otimes h_0+\ket 1\otimes h_1)\|^2 & \\
    & =\|\ket 0\otimes h_0+\ket 1\otimes fh_1\|^2  & \\
    & =\|h_0\|^2+\|fh_1\|^2 & [\textrm{by Lemma }\ref{lem:vectors in C2otimesH}] \\
    & \leq\|h_0\|^2+\|f\|^2\|h_1\|^2 & \\
    & \leq \|h_0\|^2+\|h_1\|^2=\|k\|^2 & [\textrm{since }\|f\|\leq 1]. 
\end{align*}  
Hence, $\|\Uctrl(f)\|\leq 1$, so $\Uctrl(f)\in\Ball(B(\mathbb C^2\htimes H))$.

Finally, let $f\in \Ball(B(H,K))$, so $\|f\|\leq 1$.
Instead of considering $\Uapply(f)\colon T(H)\to T(K)$, we consider its transpose $\Uapply(f)^t:B(K)\to B(H)$ (see Section \ref{sub:hs-duality}). 
For each $x\in T(H)$ and each $b\in B(K)$, we have \[\mathrm{tr}(x, \Uapply(f)^t(b))=\mathrm{tr}(\Uapply(f)(x),b)=\mathrm{tr}(fxf^\dag b)=\mathrm{tr}(xf^\dag b f)=\mathrm{tr}(x,f^\dag b f),\] making use of the cyclic property of the trace \cite[Lemma 3.4.11]{pedersen:analysisnow} in the penultimate equality. It follows that $\Uapply(f)^t=f^\dag(\cdot)f$.
Now, let $n\in\mathbb N$, let $b=[b_{ij}]\in M_n(B(K))\cong B(K^{\oplus n})$. Let $h=(h_j)\in H^{\oplus n}$. Let $k\in K^{\oplus n}$ be given by $k\eqdef(fh_j)$. Using (\ref{eq:norm inner product Hilbert sum}) and $\|f\|\leq 1$, we obtain \[\|k\|^2=\sum_{j=1}^n\|fh_j\|^2\leq\sum_{j=1}^n\|f\|^2\|h_j\|^2\leq\sum_{j=1}^n\|h_j\|^2=\|h\|^2.\] Then
\begin{align*}
    \|(\Uapply(f)^t)_n(b)h\|^2 & = \|[f^\dag b_{ij}f](h_j)\|^2 = \left\|\left(\sum_{j=1}^n f^\dag b_{ij}fh_j\right)_i\right\|^2\overset{(\ref{eq:norm inner product Hilbert sum})}{=}\sum_{i=1}^n\left\|\sum_{j=1}^n f^\dag b_{ij}fh_j\right\|^2\\
    & = \sum_{i=1}^n\left\|f^\dag \sum_{j=1}^n b_{ij}k_j\right\|^2\leq \sum_{i=1}^n\|f^\dag\|^2\left\| \sum_{j=1}^n b_{ij}k_j\right\|^2=\|f\|^2\sum_{i=1}^n\left\| \sum_{j=1}^n b_{ij}k_j\right\|^2\\
    & \leq\left\|\left(\sum_{j=1}^nb_{ij}k_j\right)_i\right\|^2 = \|bk\|^2\leq\|b\|^2\|k\|^2\leq\|b\|^2\|h\|^2.
\end{align*}
Thus $\|(\Uapply(f)^t)_n\|\leq 1$ for each $n\in\mathbb N$, so $\|\Uapply(f)^t\|_\cb\leq 1$, i.e., $\Uapply(f)^t$ is a complete contraction. It now follows from Proposition \ref{prop:hs-correspondence} that also $\Uapply(f)$ is a complete contraction, so $\|\Uapply(f)\|_\cb\leq 1$, hence $\Uapply(f)\in\Ball(CB(T(H),T(K))$.
\end{proof}

We define the Haagerup tensor product. We refer to \cite[Section 9]{effros-ruan} and \cite[Section 1.5.4]{blecher-merdy} for details. 
Let $X$, $Y$ and $Z$ be operator spaces.
Let $u:X\times Y\to Z$ be a bilinear map. For $n\in\mathbb N$, define
a bilinear map $u_{(n)}\colon M_{n}(X)\times M_{n}(Y)\to M_n(Z)$ by 
\[ (f,g)\mapsto \left[\sum_{k=1}^nu(f_{ik},g_{kj})\right]_{i,j}.\]
We say that $u$ is \emph{multiplicatively bounded} if $\sup_{n\in\mathbb N}\|u_{(n)}\|<\infty$, in which case we write $\|u\|_{\mathrm{mb}}\eqdef \sup_{n\in\mathbb N}\|u_n\|$. 
If $\|u\|_{\mathrm{mb}}\leq 1$, we call $u$ \emph{multiplicatively contractive}.
We note that in \cite{blecher-merdy} multiplicatively bounded maps and multiplicatively contractive maps are simply called completely bounded bilinear maps and completely contractive bilinear maps, respectively.

We denote the space of multiplicatively bounded bilinear maps $X\times Y\to Z$
by $MB(X\times Y,Z)$, which is a Banach space when equipped with $\|\cdot\|_\mathrm{mb}$. We can equip $MB(X\times Y,Z)$ with an operator structure by the identification $\MM_n(MB(X\times Y,Z)\cong MB(X\times Y,M_n(Z))$. 

For $n,r\in\mathbb N$ and $x=[x_{ij}]\in \MM_{n,r}(X)$ and $y=[y_{ij}]\in\MM_{r,n}(Y)$, we define $x\odot y\eqdef \left[ \sum_k x_{ik}\otimes y_{kj}\right]$ in $\MM_{n}(X\otimes Y)$.  Now, the \emph{Haagerup} tensor product $X\otimes^h Y$ of $X$ and $Y$ is defined as the completion of the algebraic tensor product $X\otimes Y$ with respect to the norm $\|\cdot\|_h$ on $\MM_n(X\otimes Y)$ defined by 
\[ \|z\|_h\eqdef \inf\{\|x\|\|y\|\colon z=x\odot y, x\in M_{n,r}(X),y\in M_{r,n}(Y), r\in\mathbb N\}.\]
If $\pi:X\times Y\to X\otimes^h Y$ denotes the bilinear map $(x,y)\mapsto x\otimes y$, then for each multiplicatively bounded bilinear map $u:X\times Y\to Z$ there is a unique completely bounded map $\tilde u:X\otimes^h Y\to Z$ such that $\tilde u\circ\pi= u$ and $\|\tilde u\|_\mathrm{cb}=\|u\|_\mathrm{mb}$. This induces a complete isometric isomorphism $MB(X\times Y,Z)\to CB(X\otimes^h Y\to Z)$.

\hagDecompose*
\begin{proof}
  The version of this proposition for completely bounded maps is proven in
  \cite[Theorem 9.4.3]{effros-ruan}. The version with the completely
  contractive maps follows as a simple corollary as follows.
  If $\varphi$ is a complete contraction, we have $\|\psi_1\|_\mathrm{cb}\|\psi_2\|_\cb=\|\varphi\|_\cb\leq 1$, hence at least one of $\|\psi_1\|_\cb$ and $\|\psi_2\|_\cb$ must be smaller than or equal to $1$. If one of them is larger than $1$, say without loss of generality $\|\psi_2\|_\cb>1$, define $\tilde\psi_1\eqdef \|\psi_2\|_\cb\psi_1$ and $\tilde\psi_2\eqdef \psi_2/\|\psi_2\|_\cb$. Then both $\tilde\psi_1$ and $\tilde\psi_2$ are complete contractions and $\varphi(x_1\otimes x_2)=\tilde\psi_1(x_1)\circ\tilde\psi_2(x)$.
  The other direction is also very easy and follows as a corollary using the fact that $\varphi$ is a complete contraction iff $\norm{\varphi}_{\cb} \leq 1.$
\end{proof}

\begin{proposition}
  The quantum switch
  \begin{align}
    & \qsw \colon B(H) \ptimes B(H) \to B(\mathbb C^2 \htimes H) \\
    & \qsw(f \otimes g) \eqdef (\ket 0 \bra 0 \htimes (fg) ) + (\ket 1 \bra 1 \htimes (gf) ) 
  \end{align}
  is a complete contraction with $\|\qsw\|_\mathrm{cb}=1$. 
\end{proposition}
\begin{proof}
Write $K\eqdef \mathbb C^2\htimes H$. Let $u:B(H)\times B(H)\to B(K)$ be the bilinear map defined by 
\[ u(f,g)=p_0\htimes fg+ p_1\htimes gf.  \]

Let $f,g\in B(H)$. 
We first show that $u$ is a joint contraction i.e., $\|u(f,g)\|\leq\|f\|\|g\|$. This might be a bit superfluous, but the essential steps are precisely the same as in proof of the `quantum' case, i.e., the proof of $u$ being a jointly complete contraction, but without all the indices appearing due to working in the `quantum' case. 

Let $k\in K$. By Lemma \ref{lem:vectors in C2otimesH} we can write $k=\ket 0\otimes h_0+\ket 1\otimes h_1$ for some $h_0,h_1\in H$. Then:
\begin{align*}
    \|u(f,g)k\|^2 & = \| (p_0\otimes fg+p_1\otimes gf)(\ket 0\otimes h_0+\ket 1\otimes h_1)\|^2 & \\
    & = \|\ket 0\otimes fgh_0+\ket 1\otimes gfh_1\|^2 & \\
    & = \|fgh_0\|^2 +\|gfh_1\|^2 & [\textrm{by Lemma }\ref{lem:vectors in C2otimesH}]\\
    & \leq \|f\|^2\|g\|^2\|h_0\|^2+\|g\|^2\|f\|^2\|h_1\|^2 & \\
    & = \|f\|^2\|g\|^2(\|h_0\|^2+\|h_1\|^2) & \\
    & = \|f\|^2\|g\|^2\|k\|^2 & [\textrm{by Lemma }\ref{lem:vectors in C2otimesH}]. 
\end{align*}
We conclude that $\|u(f,g)k\|\leq\|f\|\|g\|\|k\|$, so $\|u\|_\mathrm{jb}\leq 1$. 

For the `quantum' case, we first remark that the multiplication on $B(H)$ is a jointly complete contractive bilinear map. This follows because it is a multiplicatively contractive map \cite[1.5.4]{blecher-merdy}, and any multiplicatively contractive map is also jointly completely contractive \cite[1.5.11]{blecher-merdy}.

We furthermore recall Example \ref{ex:C-star-algebra-OSS}, which states that the OSS on $B(H)$ is provided by the identification of $\MM_n(B(H))$ with $B(H^{\oplus n})$. We will use this identification without further reference. 

Now let $[f_{ij}]\in M_n(B(H))$, $[g_{rs}]\in M_m(B(H))$. Let $k=(k_{js})_{(j,s)}\in K^{\oplus nm}.$ 
Note that $k$ is not a matrix; it has a single index $(j,s)$ that consists of a pair, because it is an element of a vector space of dimension $\dim(K)^{n\times m}$. On the other hand,  $[u(f_{ij},g_{rs})]_{(i,r),(j,s)}$ is a matrix; also its two indices $(i,r)$ and $(j,s)$ consists of pairs. 

Furthermore, by equation (\ref{eq:norm inner product Hilbert sum}), we have
$\|k\|^2=\sum_{j=1}^n\sum_{s=1}^m\|k_{js}\|^2$.
By Lemma \ref{lem:vectors in C2otimesH}, we can write $k_{js}=\ket 0\otimes h_{0,js}+\ket 1\otimes h_{1,js}$ for some $h_{i,js}\in H$, hence $\|k_{js}\|^2=\|h_{0,js}\|^2+\|h_{1,js}\|^2$. Now, if we write $h_0=(h_{0,js})$ and $h_1=(h_{1,js})$ in $H^{\oplus nm}$, then $\|h_i\|^2=\sum_{j=1}^n\sum_{s=1}^m\|h_{i,js}\|^2$ for $i=0,1$, whence 
\begin{equation}\label{eq:k-norm}
\|k\|^2=\sum_{j=1}^n\sum_{s=1}^m(\|h_{0,js}\|^2+\|h_{1,js}\|^2=\|h_0\|^2+\|h_1\|^2.
\end{equation}
Since the multiplication on $B(H)$ is a jointly completely contractive bilinear map, we have
\begin{align}
\label{ineq:multiplication-jcc-1}   \| [f_{ij}g_{rs}]\| & \leq\|[f_{ij}]\|\|[g_{rs}]\|,\\
 \label{ineq:multiplication-jcc-2}  \|[g_{rs}f_{ij}]\| & \leq \|g_{rs}\|\|f_{ij}\|.
\end{align}
Then

\begin{align*}
    \| [u(f_{ij},g_{rs})] k\|^2 & = \|[u(f_{ij},g_{rs})]_{(i,r),(j,s)}(k_{js})\|^2 & \\
    & = \left\| \left(\sum_{j=1}^n\sum_{s=1}^m u(f_{ij},g_{rs})k_{js}\right)_{(i,r)}\right\|^2 & [\text{note: norm on }K^{\oplus nm}]\\
    & =\sum_{i=1}^n\sum_{r=1}^m \left\| \sum_{j=1}^n\sum_{s=1}^m u(f_{ij},g_{rs})k_{js}\right\|^2 & [\text{by }(\ref{eq:norm inner product Hilbert sum}),\text{note: sums of norms on }K]\\
   &  =\sum_{i=1}^n\sum_{r=1}^m \left\| \sum_{j=1}^n\sum_{s=1}^m (p_0\otimes f_{ij}g_{rs}+p_1\otimes g_{rs}f_{ij})(\ket 0\otimes h_{0,js}+\ket 1\otimes h_{1,js})\right\|^2 & \\
    & =\sum_{i=1}^n\sum_{r=1}^m \left\| \sum_{j=1}^n\sum_{s=1}^m( \ket 0\otimes f_{ij}g_{rs}h_{0,js}+\ket 1\otimes g_{rs}f_{ij}h_{1,js})\right\|^2 & \\
    & =\sum_{i=1}^n\sum_{r=1}^m \left\| \ket 0\otimes\sum_{j=1}^n\sum_{s=1}^m f_{ij}g_{rs}h_{0,js}+\ket 1\otimes \sum_{j=1}^n\sum_{s=1}^m g_{rs}f_{ij}h_{1,js}\right\|^2 & \\
     & = \sum_{i=1}^n\sum_{r=1}^m \left(\left\| \sum_{j=1}^n\sum_{s=1}^m f_{ij}g_{rs}h_{0,js}\right\|^2+\left\| \sum_{j=1}^n\sum_{s=1}^m g_{rs}f_{ij}h_{1,js}\right\|^2\right) & [\text{by Lemma }\ref{lem:vectors in C2otimesH}]\\
             & = \sum_{i=1}^n\sum_{r=1}^m \left(\left\| [f_{ij}g_{rs}]_{(i,r),(j,s)}(h_{0,js}) \right\|^2+\left\| [g_{rs}f_{ij}]_{(i,r),(j,s)}(h_{1,js})\right\|^2\right) & \\
             & = \left\| [f_{ij}g_{rs}]h_0 \right\|^2+\left\| [g_{rs}f_{ij}]h_1\right\|^2 &  [\text{by }(\ref{eq:norm inner product Hilbert sum})]\\   
             & \leq \| [f_{ij}g_{rs}]\|^2\|h_0 \|^2+\| [g_{rs}f_{ij}]\|^2\|h_1\|^2 & \\
               & \leq \| [f_{ij}]\|^2\|[g_{rs}]\|^2\|h_0 \|^2+\| [g_{rs}]\|^2\|[f_{ij}]\|^2\|h_1\|^2 & [\text{by }(\ref{ineq:multiplication-jcc-1})\text{ and }(\ref{ineq:multiplication-jcc-2})]\\
               & = \|[f_{ij}]\|^2\|[g_{rs}]\|^2(\|h_0\|^2+\|h_1\|^2) & \\
               & = \|[f_{ij}]\|^2\|[g_{rs}]\|^2\|k\|^2 & [\textrm{by }(\ref{eq:k-norm})].
\end{align*}

Thus, $\|[u(x_{ij},y_{rs})]k\|\leq\|[x_{ij}]\|\|[y_{rs}]\|\|k\|$, which shows that $u$ is jointly complete contractive, i.e., $\|u\|_\mathrm{jcb}\leq 1$.

Note that $u(\id_H,\id_H)=p_0\otimes\id_H+p_1\otimes\id_H=(p_0+p_1)\otimes \id_H=\id_{\mathbb C^2}\otimes \id_H=\id_K$. Hence by choosing $f_{ii}=\id_H$ for each $i$,  $g_{rr}=\id_H$ for each $r$, $f_{ij}=0$ if $i\neq j$, and $g_{rs}=0$ if $r\neq s$, we have 
\[ u(f_{ij},g_{rs}) =\begin{cases} \id_H, & (i,r)=(j,s),\\
0, & (i,r)\neq(j,s),
\end{cases}\]
so $[u(f_{ij},g_{rs})]=\id_{K^{\oplus nm}}.$ Since the norm of the identity equals $1$, we must have $\|u\|_{\mathrm{jcb}}=1$. It follows that $\qsw=\bar u\colon B(H)\ptimes B(H)\to B(\mathbb C\htimes H)$ is a complete contraction with $\|\qsw\|_\mathrm{cb}=\|\bar u\|_\mathrm{cb}=\|u\|_{\mathrm{jcb}}=1$.
\end{proof}

\begin{proposition}
    Let $H$ be a Hilbert space, and let $u\colon B(H)\times B(H)\to B(\mathbb C^2\htimes H)$ be the bilinear map 
\[ u(f,g)= p_0\otimes fg+p_1\otimes gf.\]
Then $\|u\|_\mathrm{mb}\geq n$ for any natural number  $n\leq\dim H$. In particular, if $\dim H>1$, then $u$ is not multiplicatively contractive and if $\dim H=\infty$, then $u$ is not multiplicatively bounded. 
\end{proposition}
\begin{proof}
Since $n$ by assumption is a natural number smaller than  or equal to $\dim H$, it follows that $H$ has at least $n$ linearly independent vectors $v_1,\ldots,v_n$, which we can choose to be orthogonal and normalized. Moreover, can extend $v_1,\ldots,v_n$ to an orthonormal basis $V$ of $H$. 

For $1\leq i\leq n$, let $e_{ij}\in B(H)$ be defined as the operator satisfying $e_{ij}v_k=\delta_{jk}v_i$, for each $1\leq k\leq n$, and $e_{ij}v=0$ for each $v\in V\setminus\{v_1,\ldots,v_n\}$. Fix $i,j,k,l\in\{1,\ldots,n\}$. Then for each $1\leq m\leq n$, we have $e_{ij}e_{kl}v_m=\delta_{lm}e_{ij}v_k=\delta_{jk}\delta_{lm}v_i=\delta_{jk}e_{il}v_m$, and $e_{ij}e_{kl}v=0=\delta_{jk}e_{il}v$ for each $v\in V\setminus\{v_1,\ldots,v_n\}$. Thus, $e_{ij}e_{kl}=\delta_{jk}e_{il}$.

We now define $f\in M_n(B(H))$ by $f=[f_{ij}]_{i,j}$ be given by $f_{ij}\eqdef e_{ji}$. Let $h=(h_i)_i\in H^{\oplus n}$. By (\ref{eq:norm inner product Hilbert sum}), we have $\|h\|^2=\sum_{i=1}^n\|h_i\|^2$. Moreover, for each $i$, we have $h_i=\sum_{v\in V}\lambda_{i,v}v$ for some $\lambda_{i,v}\in\mathbb C$, hence by Parseval's identity, we have $\|h_i\|^2=\sum_{v\in V}|\lambda_{i,v}|^2$. Then 
\begin{align*}
    \|fh\|^2& =\|[f_{ij}](h_j)\|^2=\left\|\left(\sum_{j=1}^nf_{ij}h_j\right)_i\right\|^2 \overset{(\ref{eq:norm inner product Hilbert sum})]}{=} \sum_{i=1}^n \left\|\sum_{j=1}^nf_{ij}h_j\right\|^2=\sum_{i=1}^n \left\|\sum_{j=1}^ne_{ji}\sum_{v\in V}\lambda_{j,v}v\right\|^2\\
    & = \sum_{i=1}^n \left\|\sum_{j=1}^n\sum_{k=1}^n\lambda_{j,v_k}e_{ji}v_k\right\|^2= \sum_{i=1}^n \left\|\sum_{j=1}^n\sum_{k=1}^n\lambda_{j,v_k}\delta_{ik}v_j\right\|^2= \sum_{i=1}^n \left\|\sum_{j=1}^n\lambda_{j,v_i}v_j\right\|^2\overset{[\text{Parseval}]}{=}\sum_{i=1}^ n\sum_{j=1}^n|\lambda_{j,v_i}|^2\\
    & \leq \sum_{j=1}^n\sum_{v\in V}|\lambda_{j,v}|^2=\sum_{j=1}^ n\|h_j\|^2\overset{(\ref{eq:norm inner product Hilbert sum})]}{=} \|(h_j)_j\|^2=\|h\|^2,
\end{align*}
so $\|fh\|\leq\|h\|$, hence $\|f\|\leq 1$. 

Now, define $k_1,\ldots, k_n\in \mathbb C^2\htimes H$ as follows. Let $k_1\eqdef \ket 1\otimes v_1$, and $k_j\eqdef 0$ for each $j>1$. Then $\|k\|^2=\sum_{j=1}^n\|k_j\|^2=\|k_1\|=\|\ket 1\otimes v_1\|=\|\ket 1\|\|v_1\|=1$, using that the norm on $\mathbb C^2\htimes \mathbb C^n$ is a cross norm. So $k$ is a unit vector in $\mathbb C^2\htimes H$.  
Then
\begin{align*}
    \|u_{(n)}(f,f) k\|^2 & = \left\|\left[\sum_{l=1}^n u(f_{il},f_{lj})\right]_{i,j}(k_j)_j\right\|^2 = \left\|\left(\sum_{l=1}^nu(f_{il},f_{lj})k_j\right)_i\right\|^2 \overset{(\ref{eq:norm inner product Hilbert sum})}{=} \sum_{i=1}^n\left\|\sum_{j=1}^n\sum_{l=1}^nu(f_{il},f_{lj})k_j\right\|^2
    \\
    & = \sum_{i=1}^n\left\|\sum_{l=1}^nu(f_{il},f_{l1})k_1\right\|^2  = \sum_{i=1}^n\left\|\sum_{l=1}^n(p_0\otimes f_{il}f_{l1}+p_1\otimes f_{l1}f_{il})\ket 1\otimes v_1\right\|^2= \sum_{i=1}^n\left\|\sum_{l=1}^n \ket 1\otimes f_{l1}f_{il}v_1\right\|^2
    \\
   & =  \sum_{i=1}^n\left\|\ket 1\otimes \sum_{l=1}^n f_{l1}f_{il}v_1\right\|^2 \overset{[\textrm{Lemma }\ref{lem:vectors in C2otimesH}]}{=}  \sum_{i=1}^n\left\|\sum_{l=1}^n f_{l1}f_{il}v_1\right\|^2 =  \sum_{i=1}^n\left\|\sum_{l=1}^n e_{1l}e_{li}v_1\right\|^2
   =  \sum_{i=1}^n\left\|\sum_{l=1}^n e_{1i}v_1\right\|^2\\
   & =  \sum_{i=1}^n\|n e_{1i}v_1\|^2 = \sum_{i=1}^n\|n \delta_{i1}v_1\|^2 = \|nv_1\|^2 =n^2\|v_1\|^2 = n^2.\end{align*}
   Since $\|k\|=1$, it follows that $\|u_{(n)}(f,f)\|\geq n$, and since $\|f\|\leq 1$, it follows that $\|u_{(n)}\|\geq n$. Hence we must have $\|u\|_\mathrm{mb}\geq n$. The statements about $u$ not being multiplicative contractive or multiplicatively bounded follow immediately.
\end{proof}

\begin{corollary}
    Let $H$ be a Hilbert space. If $\dim H>1$, then there is no complete contraction $\qsw'\colon B(H)\hagtimes B(H)\to B(\mathbb C^2\htimes H)$ such that
    \[\qsw'(f\otimes g)=\ket 0\bra 0 \otimes fg+\ket 1\bra 1\otimes gf\] for each $f,g\in B(H)$. If $\dim H=\infty$, then there is no such $\qsw'$ that is completely bounded.
\end{corollary}

\end{document}